\documentclass{llncs} 

\newcommand\facllncs[2]{#2} 

\usepackage{ifpdf} 
\usepackage[misc,geometry]{ifsym}  
\usepackage[T1]{fontenc} 
\usepackage[utf8]{inputenc}
\usepackage{verbatim}
\usepackage[english]{babel} 
\usepackage[leqno]{amsmath}  
\usepackage{amssymb} 
 \usepackage{amsbsy,wasysym}
\usepackage{amsfonts}
\usepackage{amstext}
\usepackage{mathrsfs}
 \usepackage{algorithm}
\usepackage{latexsym}
\usepackage{array}
\usepackage{multirow}
\usepackage{fancybox}  
\usepackage{graphicx}  
\usepackage{lipsum} 
\usepackage{lineno,hyperref} 
\usepackage{color} 
\usepackage[normalem]{ulem}
\usepackage{ucs}
\usepackage{url}
\usepackage{makeidx}
\usepackage[all]{xy}
\usepackage{xspace}
\usepackage{wrapfig}
\usepackage{subfig}
\usepackage{rotating}
\usepackage{lscape}
\usepackage{epsfig}
\usepackage{float}

	\newlength{\cellH}
\newlength{\cellW}

\newcommand\newpi[2]{\hlight{#1}{NP\#\! #2}}

\newcommand\newpii[3]{\hlightNew{#1}{NP:~#2}{#3}}

\newcommand\zdnewok[2]{{#1}}

\newcommand\zdnewwok[3]{{#1}}
\newcommand\zdnewwno[3]{}

\newcommand\zdmodok[3]{\modifyok{#1}{#2}{ZD:#3}}
\newcommand\zdmoddok[4]{\modiffyok{#1}{#2}{ZD:#3}{#4}}

\newcommand\zdmodno[3]{\modifyno{#1}{#2}{ZD:#3}}
\newcommand\zdmoddno[4]{\modiffyno{#1}{#2}{ZD:#3}{#4}}

\newcommand{\modifyok}[3]{{#2}}
\newcommand{\modiffyok}[4]{{#2}}

\newcommand{\modifyno}[3]{{#1}}
\newcommand{\modiffyno}[4]{{#1}}

\newcommand\hkmodok[3]{\modifyok{#1}{#2}{HK:#3}}
\newcommand\hkmodno[3]{\modifyno{#1}{#2}{HK:#3}}

\newcommand\newno[2]{}

\newcommand\zdmoddtok[4]{\zdmoddok{#1}{#2}{#3}{#4}}
\newcommand\zdmodtok[3]{\zdmodok{#1}{#2}{#3}}

\newcommand\mor{morphism}

\newcommand\mmodel{multimodel}
\newcommand\corring{corresponding}

\newcommand\corrce{correspondence}

\newcommand\eg{e.g.}
\newcommand\ie{i.e.}

\newcommand\etal{{\em et al}}

\newcommand{\cocy}{consistency}

\newcommand\syncon{synchronization}
\newcommand\syncing{synchronizing}

\newcommand\syncer{synchronizer}
\newcommand\synced{synchronized}

\newcommand\Syncon{Synchronization}

\newcommand\fwk{framework}

\newcommand\mysubsubsection[2]{\vspace{#2}\textit{#1}}
\renewcommand\mysubsubsection[1] {\vspace{0.25ex}\noindent{\bf\em #1.}\vspace{0.ex}}
\renewcommand\mysubsubsection[2]{\par\vspace{-2ex}%
	\subsubsection{\thissec.#1 #2}
}

\DeclareMathAlphabet{\mathantt}{OT1}{antt}{li}{it}
\DeclareMathAlphabet{\mathpzc}{OT1}{pzc}{m}{it}

\newenvironment{defin}[1][]{\ifthenelse{\equal{#1}{}}{\definition}{\definition[#1]}\rm}{\enddefinition}

\newcommand\nia{\nicex{A}}

\newcommand\so{\ensuremath{\mathsf{src}}}
\newcommand\ta{\ensuremath{\mathsf{trg}}}
\renewcommand\so{\ensuremath{\mathsf{s}}} 
\renewcommand\ta{\ensuremath{\mathsf{t}}}  

\newcommand\id[1]{\opname{id}_{{#1}}}
\renewcommand\id{\opname{id}}

\newcommand\figref[1]{Fig.~\ref{#1}}

\newcommand\defref[1]{Def.~\ref{#1}}

\newcommand\sectref[1]{Sect.~\ref{#1}}

\newcommand\eqdef{\ensuremath{\stackrel{\mathrm{def}}{=}}}

\newcommand\timm{{\times}}
\newcommand\inn{{\in}}

\newcommand\ovr[1]{\ensuremath{\overline{#1}}}

\newcommand\restrict[2]{#1\left\lceil_{#2}\mathstrut\right.}
\renewcommand\restrict[2]{\ensuremath{#1{\upharpoonright}_{#2}\mathstrut}}
\newcommand\restr[2]{\restrict{#1}{#2}}

\newcommand\comprehension[2]{\ensuremath{\left\{{#1}\left|\;{#2}\right.\right\}}}
\renewcommand\comprehension[2]{\ensuremath{\left\{{#1}| \;{#2}\right\}}}  
\newcommand\compr[2]{\comprehension{#1}{#2}}

\newcommand\comprfam[2]{\left(\mathstrut{#1}\left|\;{#2}\right.\right)}

\renewcommand\comprfam[2]{(\mathstrut{#1}\left|\;{#2}\right.)}

\newcommand\noteq{\ne}

\newcommand\gracat{\catname{Graph}}

\makeatletter
  \def\fps@figure{tp}
\makeatother

\newcounter{defCounter}
\newenvironment{defin}[1][]{
\refstepcounter{defCounter}
\begin{trivlist}
\item[\hskip \labelsep {\bfseries Definition \arabic{defCounter} (#1) }]}%
{\end{trivlist}}

\newcommand\eqdef{\stackrel{\rm def}{=}}

\newcommand\eer{\datatrend{eer}}

\newcommand\corr{correspondence}

\newcommand\eg{e.g.}
\newcommand\ie{i.e.}

\newcommand\etal{{\em et al}}

\newcommand\dolan
{\mbox{$\Delta $}}

\newcommand\mathsymbol[1]{\mbox{$#1$}}
\newcommand\world
{\mathsymbol{\cal W}}
\newcommand\cmter
{\mathsymbol{\cal C}}

\newcommand\formalObj[1]{\mbox{\sffamily #1}}

\newcommand\informalObj[1]{\mbox{\textit #1}}

\newcommand\FormalInterpret[1]{\mbox{$\formalObj{i}_f$}}
\newcommand\SubstantInterpret[1]{\mbox{$\informalObj{i}_s$}}

\renewcommand\eg{e.g.}

\newcommand\categoryname[1]{\mathbf{#1}}

\newcommand\gracat{\mbox{$\categoryname{Graphs}$}}

\newcommand\corrce{correspondence}

\newcommand\ovr[1]{\mbox{$\overline{\mathstrut#1}$}}

\newcommand\id{\textsl{id}}
\renewcommand\id{\mbox{$\iota$}}

\newcommand\frar[3]{\ensuremath{#1{:}\;#2\rightarrow #3}} 

\newcommand\flar[3]{\ensuremath{#1\!:#2\leftarrow #3}} 

\newcommand\biarr[3]{\mbox{$#1\!:#2\leftrightarrow #3$}}%

\usepackage{diagrams}

\newarrowhead{bt}\blacktriangleright\blacktriangleright\blacktriangleright\blacktriangleright
\newarrowtail{bt}\blacktriangleleft\blacktriangleleft\blacktriangleleft\blacktriangleleft
\newarrowmiddle{bar}\rtbar\ltbar\dtbar\utbar
\newarrowmiddle{|}\vert\vert--
\newarrowmiddle{||}\Vert\Vert==
\newarrowmiddle{=}==\Vert\Vert
\newarrowtail{=}==\Vert\Vert
\newarrowhead{l}\langle\langle\langle\langle 
\newarrowhead{r}\rangle\rangle\rangle\rangle 
\newarrowfiller{=}==\Vert\Vert
\newarrowfiller{d}\cdot\cdot\cdot\cdot
\newarrowmiddle{x}****
\newarrowmiddle{b}\bullet\bullet\bullet\bullet
\newarrowtail{b}\bullet\bullet\bullet\bullet
\newarrowmiddle{3}\equiv\equiv\vfthree\vfthree
\newarrowfiller{3}\equiv\equiv\vfthree\vfthree
\newarrowtail{3}\equiv\equiv\vfthree\vfthree
\newarrowtail{<=}\Leftarrow\Rightarrow{\@cmex7E}{\@cmex7F}
\newarrowfiller{bold}{\bf-}{\bf-}{\bf|}{\bf|}  
\newarrowfiller{o}{\hho}{\circ}{\circ}{\circ}  
\newarrowmiddle{>}\rtla\ltla\dtla\utla
\newarrowmiddle{d}\cdot\cdot\cdot\cdot

\newarrow{To}----{>}
\newarrow{Mapsto}b---{>}
\newarrow{Mapsto}{|}---{>}
\newarrow{Dermapsto}{|}{dash}{}{dash}{>}
\newarrow{Maps}b----
\newarrow{ID}33333

\newarrow{Dashes}{}{dash}{}{dash}{}
\newarrow{Dots}.....

\newarrow{EQ}=====
\newarrow{NRelto}--+-{->}
\newarrow{Relto}--b-{->}
\newarrow{BSpanto}--b-{->} 
\newarrow{Embed}>---{->}
\newarrow{Dashto}{}{dash}{}{dash}{>} 
\newarrow{RDiagto}3333{r}
\newarrow{RDiagderto}{}{3}{}{3}{r}
\newarrow{LDiagto}3333{l}
\newarrow{LDiagderto}{}{3}{}{3}{l}
\newarrow{Mapsderto}{b}{dash}{}{dash}{>}
\newarrow{Into}C---{>}
\newarrow{Indashto}C{dash}{}{dash}{>}
\newarrow{Derinto}C{dash}{}{dash}{>}
\newarrow{Congruent}33333
\newarrow{Cover}----{blacktriangle}
\newarrow{Monic}{>}--->
\newarrow{Isoto}{>}---{triangle}
\newarrow{ISA}===={=>}
\newarrow{Isato}===={=>}
\newarrow{Doubleto}===={=>}
\newarrow{ClassicMapsto}{|}--->
\newarrow{Entail}{|}----
\newarrow{PArrow}{o}--->
\newarrow{MArrow}----{>>}
\newarrow{MPArrow}{o}---{>>}
\newarrow{Ogogoto}oooo{->}
\newarrow{Metato}--3->
\newarrow{MetaMapto}{|}-3->
\newarrow{OneToMany}+---{o}
\newarrow{HalfDashTo}{}{dash}{}-{>>}
\newarrow{DLine}=====
\newarrow{Line}-----
\newarrow{Tline}33333
\newarrow{Dashline}{}{dash}{}{dash}{}
\newarrow{Dotline}{}{o}{}{o}{}
\newarrow{Curlyto}{curlyvee}--->           
\newarrow{Bito}<--->
\newarrow{Bito}{<}---{>}
\newarrow{Bidito}{<}==={>}
\newarrow{Bitrito}{bt}333{bt}
\newarrow{Corrto}<--->
\newarrow{Dercorrto}{<}{dash}{}{dash}{>}

\newarrow{Instofto}{curlyvee}...{>}
\newarrow{Instofderto}{curlyvee}{dash}{}{dash}{->}
\newarrow{Updto}----{>}
\newarrow{Derupdto}{}{dash}{}{dash}{>}
\newarrow{Mchto}----{>}
\newarrow{Dermchto}{}{dash}{}{dash}{>}
\newarrow{Viewto}----{>->}
\newarrow{Derviewto}{}{dash}{}{dash}{>>}
\newarrow{Viewto}===={=>}
\newarrow{Hetmchto}===={=>}
\newarrow{Derviewto}{=}{}{=}{}{=>}
\newarrow{Idleto}===={=>}
\newarrow{Deridleto}{=}{}{=}{}{=>}

\newlength{\StatArrBody}
\newlength{\NodeFrameThickness}
\newcommand\acorr{\ensuremath{R}}

\newcommand\trmcat{\ensuremath{\mathbf{1}}}

\newcommand\ekk{\ensuremath{\mathpzc{k}}}
\newcommand\ebb{\ensuremath{\mathpzc{b}}}

\newcommand\trivl{\ensuremath{\mathpzc{t}}}
\renewcommand\trivl{\ensuremath{\mathfrak{t}}}
\newcommand\idl{\ensuremath{\mathfrak{id}}}

\newcommand\idlens[2]{\ensuremath{\idl_{#1}(#2)}}
\newcommand\trivlens[2]{\ensuremath{\trivl_{#1}(#2)}}

\newcommand\idlensna{\idlens{n}{\spA}}
\newcommand\idlensnb{\idlens{n}{\spB}}
\newcommand\trivlensna{\trivlens{n}{\spA}}

\newcommand\indb{\ensuremath{\mathsf{b}}}
\newcommand\inda{\ensuremath{\mathsf{v}}}

\newarrow{Updto}----{>}
\newarrow{Derupdto}{}{dash}{}{dash}{>}
\newarrow{Mchto}----{>}
\newarrow{Dermchto}{}{dash}{}{dash}{>}
\newarrow{Uto}----{>}
\newarrow{Uderto}{}{dash}{}{dash}{>}
\newarrow{Mto}----{>}
\newarrow{Mderto}{}{dash}{}{dash}{>}

\newcommand\ltilab[1]{\ensuremath{\stackrel{{#1}}{\Leftarrow}}}
\newcommand\rtilab[1]{\ensuremath{\stackrel{{#1}}{\Rightarrow}}}

\newcommand\spacenam[1]{\ensuremath{\mathbf{#1}}}
\newcommand\spX[1]{\ensuremath{\spacenam{#1}}}
\newcommand\spA{\spX{A}}
\newcommand\spB{\spX{B}}
\newcommand\spC{\spX{C}}
\newcommand\spG{\spX{G}}

\newcommand\spM{\spX{M}}

\newcommand\mspX[1]{\ensuremath{{\boldsymbol{\mathcal{#1}}}}}
\newcommand\mspA{\mspX{A}}

\newcommand\mspR{\mspX{R}}
\newcommand\mspS{\mspX{S}}

\newcommand\spaA{\ensuremath{\spacenam{A}}}
\newcommand\spaB{\ensuremath{\spacenam{B}}}

\newcommand\mspaA{\ensuremath{{\boldsymbol{\mathcal{A}}}}}

\newcommand\xob[1]{{\ensuremath{{#1^\bullet}}}}
\newcommand\xarr[1]{{\ensuremath{{#1^\RHD}}}}

\newcommand\xstar[1]{{\ensuremath{{#1^\bigstar}}}}

\newcommand\xtuple[2]{\ensuremath{({#1_1}...{#1_{#2}})}}
\newcommand\xtupleWOBrack[2]{\ensuremath{{#1_1}...{#1_{#2}}}}
\newcommand\tupleA{\xtuple{A}{n}}

\newcommand\tupleAWOB{\xtupleWOBrack{A}{n}}

\newcommand\vecfont[1]{\ensuremath{\mathbf{#1}}}
\renewcommand\vecfont[1]{\ensuremath{{\cal #1}}}

\renewcommand\nia{\vecfont{A}}

\newcommand\merge{\ensuremath{\diagopername{merge}}}

\newcommand\ppg{\ensuremath{\diagopername{ppg}}}

\newcommand\diagopername[1]{\ensuremath{\mathsf{#1}}}

\newcommand\get[1]{\mbox{\diagopername{get}$_{#1}$}} %
\newcommand\putl[1]{\mbox{\diagopername{put}$_{#1}$}} 
\newcommand\putlback[1]{\ensuremath{\overleftarrow{\putl[1]}}}
\renewcommand\get{\mbox{\diagopername{get}}} %

\renewcommand\putl{\mbox{\diagopername{put}}} 

\renewcommand\putlback{\ensuremath{\overleftarrow{\putl}}}
\renewcommand\putlback{\ensuremath{\putl^\circlearrowleft}}

\newcommand\lawname[1]{\textsf{{#1}}}
\newcommand\lawnamebr[1]{\textsf{{(#1)}}}
\renewcommand\lawnamebr[2]{\ensuremath{\mathsf{{(#1)}}_{#2}}}
\newcommand\lawnameBr[3]{\ensuremath{\mathsf{{(#1)}}_{#2}^{\mathsf{#3}}}}
\newlength{\lawgap}
\newcommand\rlarr{\ensuremath{\rightleftarrows}}
\newcommand\figri{\figref{fig:multiMetamod-uml}} 
\newcommand\figrii{\figref{fig:multimod-uml}}

\newcommand\figriv{\figref{fig:star-comp-example}}

\newcommand\diarefi{\figref{fig:ppgPatterns-simple}}
\newcommand\diarefii{\figref{fig:ppgPatterns-reflect}}

\newcommand\mmput{\ensuremath{\mathbf{put}}}
  \newcommand\mmu{\ensuremath{\mathbf{u}}}
\newcommand\elementToSingleton[1]{\ensuremath{\widehat{#1}}}
\newcommand\etosing[1]{\elementToSingleton{#1}}

\renewcommand\etosing[1]{\ensuremath{#1}}
\newcommand\niuu{\ensuremath{\mathbf{u}}}
\newcommand\thissec{\ref{sec:somesection}}
\newcommand\fakeob{\strut}

\renewcommand\newpi[2]{\hlight{#1}{NewPiece\#\! #2}}
\renewcommand\newpii[3]{\hlightNew{#1}{NewPiece\#\! #2}{#3}}

\newcommand\npmarginn[2]{\mymarginNew{NewSubsect\##1}{#2}}

\renewcommand\newpi[2]{#1} 
\renewcommand\newpii[3]{#1} 
\renewcommand\npmarginn[2]{#1} 
\renewcommand\npmarginn[2]{}

\renewcommand\xstar[1]{\Corrsx{#1}}
\renewcommand\xob[1]{\ensuremath{\mathsf{Ob}(#1)}}

\renewcommand\xarr[1]{\ensuremath{#1%
		^{\footnotesize{\RHD}}}}
	\renewcommand\xarr[1]{\ensuremath{#1%
			^{\mathbf{\rightarrow}}}}
\renewcommand\xarr[1]{\ensuremath{\mathsf{Ar}(#1)}}

\newcommand\Corrsx[1]{\ensuremath{#1^\bigstar}}
\renewcommand\Corrsx[1]{\ensuremath{\mathsf{Corr}(#1)}}

\newcommand\Corrsma{\Corrsx{\mspA}}

\newcommand\corr{\ensuremath{\mathsf{corr}}}

\newcommand\Corr{\ensuremath{\mathsf{Corr}}}
\newcommand\Corrcat{\ensuremath{\mathbf{Corr}}}
\newcommand\Corrx[1]{\Corrsx{#1}}

\newcommand\amdx[1]{\ensuremath{#1^{@}}}
\newcommand\amex[1]{\amdx{#1}}
\newcommand\ameui{\amex{u_i}}

\newcommand\reflectzero{Reflect0}

\newcommand\kputput{\lawname{KPutput}}
\newcommand\kputputlaw{\lawname{\kputput} law}
\newcommand\ppgrx[1]{\ensuremath{\ppg^{#1}}}
\newcommand\ppgr{\ppgrx{\acorr}}
\newcommand\ppgB{\ppgrx{B}}
\newcommand\ppghB{\ppgrx{\hat B}}

\newcommand\spE{\ensuremath{\spacenam{E}}}
\renewcommand\comprfam[2]{\left(\mathstrut{#1}\left|\;{#2}\right.)}
\renewcommand\comprfam[2]{\left(\mathstrut{#1}\,| \;{#2}\right)}

\newcommand\tuplenam[1]{\ensuremath{\pmb{#1}}}
\renewcommand\tuplenam[1]{\ensuremath{{#1}}}
\newcommand\tupX[1]{\ensuremath{\tuplenam{#1}}}
\newcommand\tupA{\tupX{A}}
\newcommand\tupu{\tupX{u}}

\newcommand\xconfig[1]{\ensuremath{\mathsf{#1}}}
\newcommand\Span{\xconfig{Span}}

\newcommand\compUpd[2]{\ensuremath{\mathsf{#1}^{{#2}}}}
\newcommand\seqK{\compUpd{K}{\RHD\!\!\RHD}}
\renewcommand\seqK{\compUpd{K}{\blacktriangleright\!\!\blacktriangleright}}

\newcommand\disjK[1]{\compUpd{K}{\blacktriangle{#1}}}
\newcommand\Kupd{\seqK}

\newcommand\Kdisj[1]{\disjK{#1}}

\newcommand\Kcorr{\compUpd{K}{\bigstar}}

\newcommand{\kl}{\ekk\timm \ell}
\newcommand\exIntro{$\ast$}
\renewcommand\exIntro{\textit{Running example}}
\newcommand\bad{\ensuremath{\dagger}}
\newcommand\prt{\ensuremath{\partial}}
\newcommand\prtbf{\ensuremath{\boldsymbol{\partial}}}
\renewcommand\ta{\ensuremath{\mathsf{t}}}
\renewcommand\so{\ensuremath{\mathsf{s}}}
\cornersize{.5} 
\newcommand\dbox[1]{\fbox{$#1$}}
\renewcommand\dbox[1]{\ovalbox{$#1$}}

\newcommand\xmi[1]{\ensuremath{M_{#1}}}
\newcommand\mo{\xmi{}}
\newcommand\mi{\xmi{1}}
\newcommand\mii{\xmi{2}}
\newcommand\miii{\xmi{3}}
\newcommand\xai[1]{\ensuremath{A_{#1}}}
\newcommand\yai[1]{\ensuremath{R_{#1}}}
\newcommand\ao{\xai{}}
\newcommand\ai{{\xai{1}}}
\newcommand\aii{{\xai{2}}}
\newcommand\aiii{{\xai{3}}}

\newcommand\rii{{(\yai{12})}}
\newcommand\riii{{(\yai{123})}}

\newcommand\addr{\nmf{Addr}}
\newcommand\name{\nmf{name}}

\newcommand{\Person}{\nmf{Person}}
\newcommand{\Comp}{\nmf{Company}}
\newcommand{\Company}{\nmf{Company}}
\newcommand{\Commute}{\nmf{Route}}

\newcommand{\Route}{\nmf{Route}}
\newcommand\Addr{\nmf{Addr}}
\newcommand\from{\nmf{from}}
\newcommand\lives{\nmf{livesAt}}
\newcommand\loc{\nmf{locAt}}
\newcommand\toattr{\nmf{to}}
\newcommand\Empl{\nmf{Employment}}
\newcommand{\eer}{\nmf{employer}}
\newcommand{\eee}{\nmf{employee}}
\newcommand{\etor}{\nmf{(e2r)}}
\newcommand{\ptop}{\nmf{(p2p)}}
\newcommand{\ctoc}{\nmf{(c2c)}}

\newcommand\namefont[1]{\ensuremath{\mathsf{#1}}}
\newcommand\nmf[1]{\namefont{#1}}

\newcommand{\netw}{\nmf{Network}}

\newcommand{\person}{\nmf{Person}}

\newcommand{\memb}{\nmf{memb}}
\newcommand{\ident}{\nmf{ident}}

\newcommand{\emails}{\nmf{emails}}

\newcommand{\too}{\nmf{to}}

\newcommand\hgg[1]{}

\title{Multiple Model Synchronization with Multiary Delta Lenses with Amendment and K-Putput
\thanks{This is an authors' copy of the article printed in {\em Formal Aspects of Computing 31(5): 611-640 (2019)} with multiple omissions in Sect. 7.1, which make that section practically unreadable. There are also several minor edits. 
}
}
\titlerunning{Multiary Delta Lenses}
{ 
\author{
Zinovy Diskin\inst{1}
\raisebox{0.75ex}{\scriptsize{(\Letter)}} 
Harald K\"{o}nig\inst{2}
and Mark Lawford\inst{1} 
}
\authorrunning{Diskin, K\"{o}nig, Lawford}
\institute{
 		{McMaster University, Hamilton, Canada}
 		\and 
	    {University of Applied Sciences FHDW Hannover, Germany}
	\\
    \email{diskinz@mcmaster.ca, harald.koenig@fhdw.de, lawford@mcmaster.ca}
}
}

\newcommand\facend{  
	\facllncs{\endinput}{}
} 
\newcommand\facenD{  
	\facllncs{\end{document}}{}
}

\newcommand\papertr[2]{#1}  

\begin{document}
\maketitle
\begin{abstract}
	Multiple (more than 2) model synchronization is ubiquitous and important for model driven engineering, but its theoretical underpinning gained much less attention than the binary case. Specifically, the latter was extensively studied by the bx community in the framework of algebraic models for update propagation called lenses. Now we make a step to restore the balance and propose a notion of multiary delta lens. Besides multiarity, our lenses feature {\em reflective} updates, when consistency restoration requires some amendment of the update that violated consistency. We emphasize the importance of various ways of lens composition for practical applications of the framework, and prove several composition results.   
\end{abstract}
\pagestyle{headings}
\section{Introduction} 


\newpii{
Modelling normally results in a set of inter-related models presenting different views of a single system at different stages of development. The former differentiation is usually referred to as ``horizontal'' (different views on the same abstraction level) and the latter as ``vertical'' (different abstraction levels beginning from the most general requirements down to design and further on to implementation). A typical modelling environment in a complex project is thus a collection of models (we will call them {\em local}) inter-related and inter-dependant along and across the horizontal and the vertical dimensions of the network. We will call the entire collection a {\em multimodel}, and refer to its component as to {\em local} models.  
\\
The system integrating local models can exist either materially (\eg, with UML modelling, a single UML model whose views are specified by UML diagrams, is physically stored by the UML tool) or virtually (\eg, several databases integrated into a federal database), or in a mixed way (\eg, in a complex modelling environment encompassing several UML models).  Irrespective of the type of integration (material, virtual, mixed), the most fundamental property of a multimodel is its  {\em global}, or {\em joint, consistency}: if local models do not contradict each other in their viewing of the system, then at least one system satisfying {\em all} local models exists; otherwise, we say local models are (globally) inconsistent.
}{1}{-8cm}

If one of the local models changes and their joint consistency is violated, the related models should also be changed to restore consistency. This task of model \syncon\ 
is obviously of paramount importance for MDE, but its theoretical underpinning is inherently difficult 
and reliable automatic synchronization solutions are rare in practice. 
Much theoretical work partially supported by implementation has been done for the binary case (synchronizing two models) by the bidirectional transformation community (bx), 
specifically, by its TGG sub-community, see, \eg, \cite{hermann2012:concurrent}), and  the {\em delta lens}  sub-community on a more abstract level (delta lenses \cite{me-models11} can be seen as an abstract algebraic interface to TGG based \syncon\ \cite{frank-models11}). 
However, disappointedly for practical applications, the case of multiary \syncon\ (the number of models to be \synced\ is $n > 2$) gained much less attention---cf. the energetic call to the community in a recent Stevens' paper \cite{stevens-models17}.  

The context underlying bx is model transformation, in which one model in the pair is considered as a transform of the other even though updates are propagated in both directions (so called round-tripping). 
Once we go beyond $n=2$, we switch to a more general context of inter-model relations beyond model-to-model transformations. Such situations have been studied in the context of multiview system consistency, see surveys \cite{marsha-isse12,DBLP:journals/corr/KnappM16}, but rarely in the context of an accurate formal basis for update propagation. 
\newpii{
A notable exception is work by Trollmann and Albayrak  \cite{frankT-icmt15,frankT-icmt16,frankT-icmt17}. In the first of these papers, they specify a grammar-based engine for generating consistent multimodels of arbitrary arity $n\ge 2$, with the case $n=2$ being managed by TGG and truly multiary cases $n\ge 3$ are uniformly managed by what they call Graph-Diagram Grammars, GDG.  In paper \cite{frankT-icmt16} they use GDG for building a multiary change propagation \fwk, which is close in its spirit to our \fwk\ developed in the paper but is much more concrete --- we will provide a detailed comparison in the Related work section. 
Roughly, our \fwk\ of multiary delta lenses developed in the paper is to GDG-based update propagation as binary symmetric delta lenses are to TGG-based update propagation, $\mathsf{mxLens/GDG} \approx \mathsf{bxLens/TGG}$, where we refer to multiary update propagation as {\em mx} (contrasting it to binary bx). The latter relationship is described in  \cite{frank-models11}: binary delta lenses appear as an abstract algebraic interface to TGG-based change propagation; at some stage, we want to achieve similar results for mx-lenses and GDG (but not in this paper). 
}{1}{-6cm}
%

 Our contributions to mx are as follows. 
 We show with a simple example (\sectref{sec:example}) an important special feature of multiview modelling: consistency restoration may require not only update propagation to other models but the very update created inconsistency should itself be amended; thus, update propagation should, in general, be {\em reflective}  (even for the case of a two-view system).   
Motivated by the example, in \sectref{sec:mmod} we formally define the notion of a multimodel, and then in \sectref{sec:mlenses}, give a formal definition of a {\em multiary} (symmetric) lens with amendment and state the basic algebraic laws such lenses must satisfy. Importantly, we have a special \lawname{KPutput} law that requires compatibility of update propagation with update composition for a restricted class \Kupd\ of composable update pairs.%

Our major results 
are about lens composition. In \sectref{sec:lego}, we define several operations  over lenses, which produce complex lenses from simple ones: we first consider two forms of parallel composition  in \sectref{sec:lego-para}, and then two forms of sequential composition in Sections \ref{sec:lego-star} and \ref{sec:lego-spans2lens}. 
Specifically, the construct of composing an $n$-tuple of asymmetric binary lenses sharing the same source into a symmetric $n$-ary lens gives a solution to the problem of building mx \syncon\ via bx discussed by Stevens in \cite{stevens-models17}. 

We consider lens composition results crucially important for practical application of the \fwk. If a tool builder has implemented a library of elementary \syncon\ modules based on lenses and, hence, ensuring basic laws for change propagation, then a complex module assembled from elementary lenses will automatically be a lens and thus also enjoys the basic laws. This allows the developer to avoid additional integration testing, which can essentially reduce the cost of \syncon\ software.   

The paper is an essential extension of our FASE'18 paper \cite{me-fase18}. The main additions are i) a new section motivating our design choices, ii) a constrained Putput law (\kputput) and its thorough discussion, including a corresponding extension of the running example, iii) a counterexample showing that invertibility is not preserved by  star composition, iv) two types of parallel composition of multiary lenses, v) Related Work and Future Work sections are essentially extended, particularly, an important subsection about multimodel updates including correspondence updates (categorification) is added.

\section{Background: Design choices for the paper}\label{sec:backgr}

\renewcommand\thissec{\ref{sec:backgr}}
\facllncs{
\renewcommand\mysubsubsection[2]{
	\subsubsection{#2}}
}{}  

In this section, we discuss our design choices for the paper: why we need multiarity, amendments, K-Putput, and why, although we recognize limitations of a \fwk\ only dealing with non-concurrent update scenarios, we still develop their accurate algebraic model in the paper.

\mysubsubsection{1}{Why multiary lenses.} 
 Consider, for simplicity, three models, \ai, \aii, and \aiii, working together (\ie, being models of the same integral system) and being in sync at some moment.  Then one of the models, say, \ai, is updated to state $A'_1$, 
and consistency is violated. To restore consistency, the two other models are to be changed accordingly and we say that the update of model \ai\ is propagated to \aii\ and \aiii. Thus, consistency restoration amounts to having three pairs of propagation operations, $(\ppg_{ij},\ppg_{ji})$, with operation $\ppg_{ij}$ propagating updates of model $A_i$ to model $A_j$, $i\noteq j$, $i,j=1,2,3$.  It may seem that the \syncon\ problem can be managed by building three binary lenses $\ell_{ij}$ -- one lens per a pair of models. 


\begin{wrapfigure}{R}{0.33\textwidth}
\vspace{-4ex}
\centering
    \includegraphics[width=0.25\textwidth]%
                         {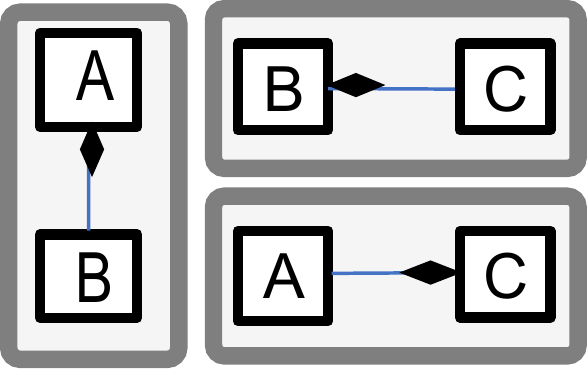}
\caption{Three jointly inconsistent class diagrams \label{fig:tri-example}}
\vspace{-4ex}
\end{wrapfigure}

However, when three models work together, their consistency is a ternary relation often irreducible to binary consistency relations. Figure \ref{fig:tri-example} presents a simple example: three class diagrams shown in the figure are pairwise consistent while the whole triple is obviously inconsistent, \ie, violates a class diagram metamodel constraint (which prohibits composition cycles).

A binary lens can be seen as a couple of Mealy machines (we write  $A_1\rlarr A_2$) sharing a state space (say, $R_{12}$). A ternary lens synchronizing a triple of models can also be seen as a triple of couples of Mealy machines ($A_1\rlarr A_2$, $A_1\rlarr A_3$, $A_2\rlarr A_3$), but they share the same space $R_{123}$ and hence mutually dependant on each other (they would be independent if each couple $A_i\rlarr A_j$would have its own space $R_{ij}$). 
Thus, multiple model \syncon\ is, in general, irreducible to chains of binary lenses and needs a new notion of a multiary lens.  

\mysubsubsection{2}{Why amendments.} 
Getting back to the example, suppose that the updated state $A_1'$ goes beyond the projection of all jointly consistent states to the space $A_1$, and hence consistency cannot be restored with the first model being in state $A'_1$. When we work with two models, such cases could be a priori prohibited by modifying the corresponding metamodel $M_1$ defining the model space so that state $A'_1$ would violate $M_1$. When the number of models to sync grows,  it seems more convenient and realistic to keep $M_1$ more flexible and admitting $A'_1$ but, instead, when the \syncer\ is restoring consistency of all models, an amendment of state $A'_1$ to a state $A''_1$ is also allowed. Thus, update propagation works {\em reflectively} so that not only other models are changed but the initiating update from $A_1$ to $A'_1$ is itself amended and model $A'_1$ is changed to $A''_1$. 
Moreover, in \sectref{sec:reflect} we will consider examples of situations when if even state $A'_1$ can be \synced, a slight amendment still appears to be a better \syncon\ policy. For example, if consistency restoration with $A'_1$ kept unchanged requires deletions in other models, while amending $A'_1$ to $A''_1$ allows to restore consistency by using additions only, then the update policy with amendments may be preferable -- as a rule, additions are preferable to deletions. 

Of course, allowing for amendments may open the Pandora box of pathological \syncon\ scenarios, \eg, we can restore consistency by rolling back the original update and setting $A''_1=A_1$. We would like to exclude such solutions and bound amendments to work like completions of the updates rather than corrections and thus disallow any sort of ``undoing''. To achieve this, we introduce a binary relation \Kupd\ of update compatibility: if an update \frar{u}{A_1}{A'_1} is followed by update \frar{v}{A'_1}{A_1''} and $u\Kupd\,v$, then $v$ does not undo anything done by $u$. Then we require that the original update and its amendment be \Kupd-related. (Relation \Kupd\ and its formal properties are discussed in \sectref{sec:modelspaces}.)

\mysubsubsection{3}{Why K-Putput.} 
An important and desired property of update propagation is its compatibility with update composition. If $u$ and $v$ are sequentially composable updates, and \putl\ is an update propagation operation, then its compositionality means  $\putl(u;v)=\putl(u);\putl(v)$%
\footnote{to make this formula precise, some indexes are needed, but we have omitted them}
 ---hence, the name Putput\ for the law.
There are other equational laws imposed on propagation operations in the lens \fwk\  to guarantee desired \syncon\ properties and exclude unwanted scenarios. Amongst them, Putput is the most controversial: 
{Putput} without restrictions does not hold while finding an appropriate guarding condition -- not too narrow to be practically usable and not too wide to ensure compositionality -- has been elusive (cf. 
\cite{foster07%
	,bpierce-popl11%
	,me-jot11,jr-bx12,me-jss15%
	 ,me-bx17%
}).
The practical importance of Putput follows from the possibilities of optimizing update propagation it opens: if Putput holds, instead of executing two propagations, the engine can executes just one. Moreover, before execution, the engine can optimize the procedure by preprocessing  the composed update $u;v$ and, if possible, converting it into an equivalent but easier manageable form (something like query optimization performed by database engines).
  
A preliminary idea of a constrained Putput is discussed in \cite{jr-bx12} under the name of a {\em monotonic} Putput: compositionality is only required for two consecutive deletions or two consecutive insertions (hence the term monotonic), which is obviously a too strong filter for typical practical applications. 
The idea of constraining Putput based on a compatibility relations over consecutive updates, \eg, relation \Kupd\ above, is much more flexible, and gives the name K-Putput for the law. It was proposed by Orejas \etal\ in \cite{orejas-bx13} for the binary case without amendment and  intermodel correspondences (more accurately, with trivial correspondences being just pairs of models), and we adapt it for the case of full-fledge lenses with general correspondences and amendments. We consider our integration of \Kupd-constrained Putput and amendments to be an important step towards making the lens formalism more usable and adaptable for practical tasks.  
%
%


\mysubsubsection{4}{Why non-concurrent \syncon.} 
In the paper, we will  consider consistency violation caused by a change of only one model, and thus consistency is restored by propagating only one update, while in practice we often deal with several models changing concurrently. If these updates are independent, the case can be covered with one-update propagation \fwk\ using interleaving, but if concurrent updates are in conflict, consistency restoration needs a conflict resolution operation (based on some policy) and goes beyond the \fwk\ we will develop in the paper. One reason for this is technical difficulties of building lenses with concurrency -- we need to specify reasonable equational laws regulating conflict resolution and its interaction with update composition. It would be a new stage in the development of the lens algebra. 

Another reason is that the case of one-update propagation is still practically interesting and covers a broad class of scenarios -- consider a UML model developed by a software engineer. Indeed, different UML diagrams are just different views of a single UML model maintained by the tool, and when the engineer changes one of the diagrams, the change is propagated to the model and then the changed model is projected to other diagrams. Our construct of star-composition of lenses (see \sectref{sec:lego-star} and \figref{fig:starComposition}) models exactly this scenario. Also, if a UML model is being developed by different teams concurrently, team members often agree about an interleaving discipline of making possibly conflicting changes, and the one-update \fwk\ is again useful. Finally, if concurrent changes are a priori known to be independent, this is well modelled by our construct of parallel composition of one-update propagating lenses (see \sectref{sec:lego-para}).

\mysubsubsection{5}{Lens terminology.}
The domain of change propagation is inherently complicated and difficult to model formally. The lens \fwk\ that approaches the task is still under development and even the basic concepts are not entirely settled and co-exist in several versions (\eg, there are {\em strong} and {\em weak invertibility}, several versions of Putput, and different names for the same property of propagating idle updates to idle updates). This results in a diverse collection of different types of lenses, each of which has a subtype of so called {\em well-behaved (wb)} lenses (actually, several such as the notion of being wb varies), which are further branched into different notions of {\em very wb} lenses depending on the Putput version accepted. 
The diversity of lens types and their properties, on the one hand, and our goal to provide accurate formal statements, on the other hand, would lead  to overly wordy formulations. To make them more compact, we use two bracket conventions.   

Square brackets. 
If we say {\em A wb lens is called {\em [weakly]} invertible, when...}, we mean that using adjective {\em weakly} is optional for this paper as the only type of invertibility we consider is the weak invertibility. However, we need to mention it because the binary lens notion that we generalize in our multiary lens notion, is the weak invertibility to be distinguished from the strong one.  Thus, we mention `weakly' in square brackets for the first time and then say just `invertible' (except in a summarizing result like a theorem, in which we again mention [weakly]).   

Round brackets. If a theorem reads {\em A span of (very) wb asymmetric lenses ... gives rise to a (very) wb symmetric lens ...} it mean that actually we have two versions of the theorem: one is for wb lenses, and the other is for very wb lenses.

\section{Example}
\label{sec:example}
We will consider a simple example motivating our \fwk. The formal constructs constituting the multiary delta lesn \fwk\  will be illustrated with the example (or its fragments) and referred to as  {\em Running example}. \newpi{Although the lens \fwk\ is formal, the running example instantiating it, will be presented semi-formally: we will try to be precise enough, but an accurate formalization would require the machinery of graphs with diagram predicates and partial graph morphisms as described in our paper \cite{me-ecmfa17}, and we do not want to overload this paper with formalities even more.}{1}

\subsection{A Multimodel to Play With}
Suppose two data sources, whose schemas (we say metamodels) are shown in \figri\ as class diagrams \mi\ and \mii\  that record  
 employment. The first source is interested in employment of people living in downtown, the second one is focused on software companies  and their recently graduated employees. 
 In general, population of classes \Person\ and \Comp\ in the two sources can be different -- they can even be disjoint, but if a recently graduated downtowner works for a software company, her appearance in both databases is very likely. 
 
 \begin{figure}[t]
\centering
    \includegraphics[
                    width=1.\textwidth%
                 ]%
                 {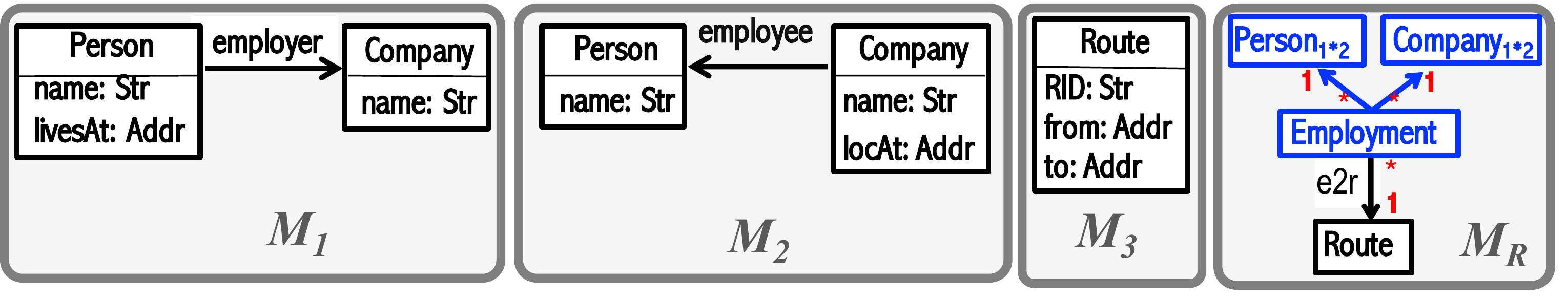}
\caption{Multi-Metamodel in UML\label{fig:multiMetamod-uml}}
\end{figure}
 Now suppose there is an agency investigating traffic problems, which maintains its own data on commuting routes between addresses  as shown by schema \miii. These data should be synchronized with commuting data provided by the first two sources and   computable by an obvious relational join over \mi\ and \mii: {roughly, the agency keeps traceability between the set of employment records and the \corring\ commuting routes and requires their \from\ and \toattr\ attributes to be synchronized  --- below we will specify this condition in detail and explain the forth metamodel $M_R$ \newpi{(see specification $(R_{123})$ on the next page)}{2}. In addition, the agency supervises consistency of the two sources and requires that if they both know a person $p$ and a company $c$, then they must agree on the employment record $(p,c)$: it is either stored by both or by neither of the sources. For this \syncon, it is assumed that persons and companies are globally identified by their names.  Thus, a triple of data sets (we will say models) \ai, \aii, \aiii, instantiating the respective metamodels, can be either consistent (if the constraints described above are satisfied) or inconsistent (if they aren't). In the latter case, we normally want to change some or all models to restore consistency. 
We will call a collection of models to be kept in sync a {\em multimodel}. 

To specify constraints for multimodels in an accurate way, we need an accurate notation.  If \ao\ is a model instantiating metamodel \mo\ and $X$ is a class in \mo, we write $X^\ao$ for the set of objects instantiating $X$ in $A$. Similarly, if \biarr{r}{X_1}{X_2} is an association in \mo, we write $r^\ao$ for the corresponding binary relation over $X_1^{\ao}\times X_2^{\ao}$.%
\footnote{In general, an association $r$ is interpreted as a multirelation $r^A$, but if the constraint [unique] is declared in the metamodel, then $r^A$ must be a relation. We assume all associations in our metamodels are declared to be [unique] by default.} 
For example, \figrii\ presents a simple model \ai\ instantiating \mi\ with $\Person^\ai=\{p_1,p_1'\}$, $\Company^\ai=\{c_1\}$,  $\eer^\ai=\{(p_1,c_1)\}$, and similarly for  attributes, \eg, 
$$\lives^\ai=\{(p_1,a1), (p'_1, a1)\}\subset \Person^\ai \timm\Addr$$ ($\lives^\ai$ and also $\name^\ai$ are assumed to be functions and $\Addr$ is the (model-independent) set of all possible addresses). Two other boxes present models $A_2$ and $A_3$ instantiating metamodels $M_2$ and $M_3$ resp.; we will discuss the rightmost box $R$ later.    
The triple $(\ai, \aii, \aiii)$ is a (state of the) multimodel over the multimetamodel $(\mi,\mii,\miii)$, and we say it is {\em consistent} if the constraints specified below are satisfied. 

Constraint  (C1) specifies mutual consistency of models $A_1$ and $A_2$  in the sense described above:
 
 
\smallskip 
\label{pg:constr}
\facllncs{}{
\noindent \begin{tabular}{l@{\qquad}l}
\multirow{2}{*}{(C1)} & 
}
if $p\in\Person^\ai\cap \Person^\aii$ and $c\in  \Comp^\ai\cap \Comp^\aii$ 
\facllncs{}{
	\\ & 
then $(p,c)\in \eer^\ai$ iff $(c,p)\in \eee^\aii$
\end{tabular}
}

\newcommand\derRoute{\ensuremath{{
			\eer^\ai\cup{(\eee^\aii)}^{-1} 
}}}
\newcommand\aiandii{\ensuremath{{\ai{*}\aii}}}

\medskip
Our other constraints specify consistency between the agency's data on commuting routes and the two data sources. 
We first assume a new piece of data that relates models: a relation $\etor^R$ whose domain is the integral set of employment records $\Empl^\aiandii$ as specified below in \riii: 


\smallskip \label{pg:constr123}
\noindent \begin{tabular}{l@{\qquad}l}
	\multirow{2}{*}{\riii} & 
	$\etor^R \subset \Empl^\aiandii\times\Route^\aiii$
	\\ & where  
	$\Empl^\aiandii\eqdef\derRoute$
\end{tabular}


\smallskip
\noindent \zdnewok{This relation is described in the metamodel $M_R$ in \figri, where blue boxes denote (derived) classes whose instantiation is to be automatically computed (as specified above) rather than is given by the user.}{} 

To simplify presentation, we assume that all employees commute rather than work from home, which means that relation $\etor^R$ is left-total. Another simplifying assumption is that each employment record maps to exactly one commuting route, hence, relation $\etor^R$ is a single-valued mapping. Finally, several people living at the same address may work for the same company, which leads to  different employment records mapped to the same route and thus injectivity is not required. Hence, we have the following intermodel constraint:

\smallskip
\noindent \begin{tabular}{l@{\qquad}l}
	(C2) &  relation $\etor^R$ is a total function 
	${\Empl^\aiandii}\longrightarrow{\Route^\aiii}$
\end{tabular} 

\smallskip
\noindent which is specified in the metamodel $M_R$ in \figri\ by the corresponding multiplicities.  However, the metamodel $M_R$ as shown in \figri, is still an incomplete specification of correspondences between models. 

For each employment record $e\in\Empl^\aiandii$, there are defined the \corring\ person $$p=e.\eee^\aiandii\in\Person^\aiandii\eqdef\Person^\ai\cup\Person^\aii$$ with attribute 
\[
p.\lives^\aiandii=\left\{
\begin{array}{ll}
p.\lives^\ai & \mbox{if $p\in\Person^\ai$,}
\\ [1ex]
\bot & \mbox{otherwise.}
\end{array}  
\right.
\]
We thus assume that the domain \Addr\ is extended with a bottom value/null $\bot$. 
Similarly, we define set $\Company^\aiandii$ with attribute $\loc^\aiandii$ and the \corring\ ``end'' $e.\eer^\aiandii \in \Company^\aiandii$ for any employment record $e$. The latter is of special interest for the agency if both addresses,
$$e.\eee^\aiandii.\lives^\aiandii\mbox{ and } e.\eer^\aiandii.\loc^\aiandii,$$ 
are defined and thus define a certain commuting route to be consistent with its image $e.\etor^R$ in $\Route^\aiii$.

More generally, the consistency of the two sets of commuting routes: that one derived from $A_1$ and $A_2$, and that one stored in $A_3$, can be specified as follows:  

\smallskip
\noindent \begin{tabular}{l@{\qquad}ll}
\multirow{2}{*}{(C3)}  & 
$e.\eee^\aiandii.\lives^\aiandii$ &$\le e.\etor^R.\from^\aiii$
\\
& $e.\eer^\aiandii.\loc^\aiandii$ &$\le e.\etor^R.\toattr^\aiii$ 
\end{tabular}

\smallskip
\noindent where inequality $a_1\le a_2$ holds iff both values are certain and $a_1=a_2$, or both values are nulls, or $a_1$ is a null while $a_2$ is certain.  

Now it is easy to see that multimodel $(A_1, A_2, A_3, R)$ in \figrii\ is ``two-times'' inconsistent: (C1) is violated as both \ai\ and \aii\ know Mary and IBM, and (IBM,Mary)$\in$ $\eee^\aii$ but (Mary, IBM)$\notin\eer^\ai$, and (C2) is violated as \aii\ and $R$ show an employment record $e_2$ not mapped to a route in $\Route^\aiii$. Note also that if we map this record to the route \#1 to fix (C2), then constraint (C3) will be violated. We will discuss consistency restoration in the next subsection, but first we need to finish our discussion of intermodel correspondences. 

\begin{figure}[t]
\centering
    \includegraphics[
                    width=1.025\textwidth,%
                 ]%
                 {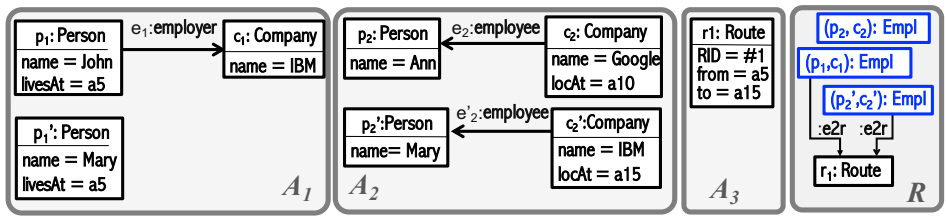}
\caption{Multimodel $\nia^\dagger=(A_1, A_2, A_3, R)$ with $\ptop^R=\{(p_1',p_2')\}$, $\ctoc^R=\{(c_1,c_2')\}$, and $\etor^R=\{(e_1,r_1), ({e'_2}^{-1}, r_1)\}$  
	\label{fig:multimod-uml}}
\end{figure}

Note that correspondences between models have more data than specified in \riii.  Indeed, classes $\Person^\ai$ and $\Person^\aii$ are interrelated by a correspondence  linking persons with the same name, and similarly for \Company\ so that we have two partial injections: 

\smallskip
\label{pg:constr12}
\noindent \begin{tabular}{l@{\qquad}l}
	\rii &  $\ptop^R\subset \Person^\ai \timm \Person^\aii$, 
	        $\ctoc^R\subset \Company^\ai \timm \Company^\aii$,
\end{tabular} 

\smallskip\noindent which are not shown in metamodel $M_R$ (by purely technical reasons of keeping the figure compact and fitting in the page width). 
These correspondence links (we will write corr-links) may be implicit as they can always be restored using names as keys. In contrast, relation \etor\ is not determined by the component model states $A_1,A_2,A_3$ and is an independent piece of data. 
\zdnewwok{Importantly, for given models $A_{1,2,3}$, there may be several different correspondence mappings $\etor^R$ satisfying the constraints. For example, if there are several people living at the same address and working for the same company, all employment record can be mapped to the same route or to several different routes depending on how much carpooling is used.}{}{-2em} 
In fact, multiplicity of possible corr-specifications is a general story; it can happen for relations \ptop\ and \ctoc\ as well if person and company names are not entirely reliable keys. 
Then we need a separate procedure of model matching or alignment that has to establish, \eg,  whether objects $p'_1\in\Person^\ai$ and $p'_2\in\Person^\aii$ both named Mary represent the same real world object. 
	Constraints we declared above implicitly involve corr-links, \eg, formula for (C1)  is a syntactic sugar for  the following formal statement: if $(p_1, p_2)\in \ptop^R$ and 
	$(c_1, c_2)\in\ctoc^R$ with $p_i\inn\Person^{A_i}$, $c_i\inn\Company^{A_i}$ ($i=1,2$),
then the following holds: 
$(p_1,c_1)\in \eer^\ai$ iff $(c_2,p_2)\in \eee^\aii$. A precise formal account of this discussion can be found in \cite{me-ecmfa17}.

Thus, a multimodel is actually a tuple $\nia=(\ai, \aii, \aiii, R)$ where $R$ is a \corrce\ specification (which, in our example, is a collection of correspondence relations \riii, \rii\ over sets involved). 
Consistency of a multimodel is a property of the entire 4-tuple \nia\ rather than its 3-tuple carrier $(\ai, \aii, \aiii)$.


\subsection{\Syncon\ via Update Propagation}
 There are several ways to restore consistency of the multimodel in \figrii\ w.r.t. constraint (C1). We may 
 delete Mary from \ai, or delete her employment with IBM from \aii, or even delete IBM from \aii. We can also change Mary's employment from IBM to Google, which will restore (C1) as \ai\ does not know Google. Similarly, we can delete John's record from \ai\ and then Mary's employment with IBM in \aii\ would not violate (C1).  As the number of constraints and the elements they involve increase, the number of consistency restoration variants grows fast.   

The range of possibilities can be essentially decreased if we take into account the history of creating in\cocy\ and consider not only an inconsistent state $\nia^\bad$ but update \frar{u}{\nia}{\nia^\bad} that created it (assuming that \nia\ is consistent). For example, suppose  that initially model \ai\ contained record (Mary, IBM) (and \aiii\ contained (a1, a15)-commute), and the inconsistency appears after Mary's employment with IBM was deleted in \ai. Then it's reasonable to restore consistency by deleting this employment record in \aii\ too; we say that deletion was propagated from \ai\ to \aii. 
If the inconsistency appears after adding (IBM, Mary)-employment to \aii, then it's reasonable to restore consistency by adding such a record to \ai. Although propagating deletions/additions to deletions/additions is typical, there are non-monotonic cases too. Let us assume that Mary and John  are spouses and live at the same address, and that IBM follows an exotic policy  prohibiting spouses to work together. Then we can interpret addition of (IBM, Mary)-record to \aii\ as swapping of the family member working for IBM, and then (John, IBM) is to be deleted from \ai. 

\zdnewwok{
Now let's consider how updates to and from model \aiii\ may be propagated.  As mentioned above, traceability/correspondence links play a crucial role here. If additions to  \ai\ or \aii\ create a new commute, the latter has to be added to \aiii\ (together with its corr-links) due to constraints (C2) and (C3). In contrast, if a new route is added to \aiii, we may change nothing in $A_{1,2}$ as (C2) does not require surjectivity of \etor\ (but further in the paper we will consider a more intricate policy). If a route is deleted from \aiii, and it is traced via $\etor^R$ to one or several corresponding employments in 
$\Empl^\aiandii$, then they are either deleted too, or perhaps remapped to other routes with the same \from-\toattr pair of attributes if such exist. Similarly, deletions in $\Empl^\aiandii$ may (but not necessarily) lead to the \corring\ deletions in $\Route^\aiii$ depending on the mapping $\etor^R$.  Finally, updating addresses in $A_1$ or $A_2$  is propagated to the \corring\ updates of \from\ and \toattr\ attributes in $A_3$ to satisfy constraint (C3); similarly for attribute updates in $A_3$. 
}{slightly updated}{-5em}

Clearly, many of the propagation policies above although formally correct, may contradict the real world changes and hence should be corrected, but this is a common problem of a majority of automatic \syncon\ approaches, which have to make guesses in order to resolve non-determinism inherent in consistency restoration. 

\subsection{Reflective Update Propagation}\label{sec:reflect}

 An important feature of update propagation scenarios above is that consistency could be restored without changing the model whose update caused inconsistency. However, this is not always desirable. Suppose again that violation of constraint (C1) in multimodel in \figrii\ was caused by adding a new person Mary to \ai, \eg, as a result of Mary's moving to downtown. Now both models know both Mary and IBM, and thus either employment record (Mary, IBM) is to be added to \ai, or record (IBM, Mary) is to be removed from \aii.  Either of the variants is possible, but in our context, adding (Mary, IBM) to \ai\ seems more likely and less specific than deletion (IBM, Mary) from \aii. Indeed, if Mary has just moved to downtown, the data source \ai\ simply may not have completed her record yet. Deletion (IBM, Mary) from
 \aii\ seems to be a different event unless there are strong causal dependencies between moving to downtown and working for IBM. Thus, an update policy that would keep \aii\ unchanged but amend addition of  Mary to \ai\ with further automatic adding her employment for IBM (as per model \aii) seems reasonable. This means that updates can be reflectively propagated (we also say self-propagated). 
 
 Of course, self-propagation does not necessarily mean non-propagation to other directions. Consider the following case: model \ai\ initially only contains (John, IBM) record and is consistent with \aii\ shown in \figrii. Then record (Mary, Google) was added to \ai, which thus became inconsistent with \aii. To restore consistency, (Mary, Google) is to be added to \aii\ (the update is propagated from \ai\ to \aii) and (Mary, IBM) to be added to \ai\ as discussed above (\ie, addition of (Mary, Google) is both amended and propagated). 
 \newpii{Note, however, that in contrast to the previous case, now deletion of the record (IBM, Mary) from \aii\  looks like an equally reasonable scenario of Mary changing her employer. Thus, even for the simple case above, and the more complex cases of model interaction, the choice of the update policy (only amend, only propagate, or both) depends on the context,  heuristics, and tuning the policy to practice. 
 }{1}{-5em}

\newpii{A typical situation that needs an amendment facility is when the changes in interacting models have different granularity. With our simple running example, we can illustrate the point in the following (rather artificial) way. Suppose, again, that the record of Mary working for Google, and her address unknown (\ie, Mary.$\lives=\bot$) is added to model \ai, and propagated to \aii\ as discussed above. Suppose that Google has a strict policy of only hiring those recent graduates who live on Bloor Street in Toronto downtown. Then in Mary's address record, all fields could (and should!) be made certain besides the street number. Hence, adding Mary's employment to model \ai\ should be amended with extending her address with data imposed by model \aii. 
}{2}{-10em}
\newpii{	
For a more realistic example, consider model \aii\ specifying a complex engineering project $P$ in the process of elaboration, while model \ai\ gives its very abstract view -- the budget $B$ of the project. If the budget changes from $B$ to $B'$, the project should also be changed, but it is very likely that the budget of the changed project $P'$ would be $B''$ rather than exactly $B'$. A more general and formal description of this \syncon\ schema can be found in \cite{me-bx17}.
}{2 cont'd}{-2em}
 %


\subsection{General schema}
\facllncs{
\cellW=7.25ex%
\cellH=5.0ex
\begin{figure}
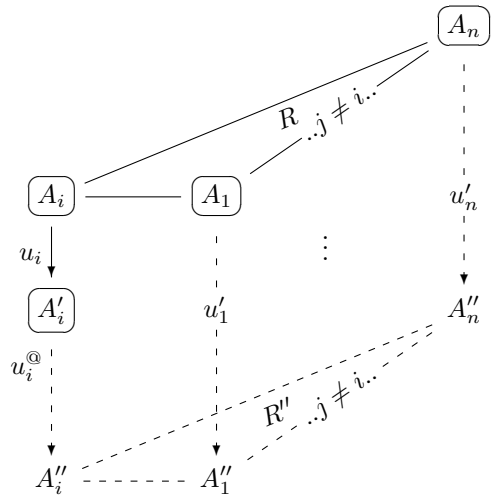

	\centering
\begin{tabular}{c} 
\begin{diagram}[w=1.\cellW,h=1.\cellH] 
&& &&&\dbox{A_n}
\\  
&&&&
\ruLine(5,3)>{\qquad R}
\ruLine(3,3)~{..j\noteq i..}
& 
\\ 
&&&\fakeob&&
\\  
\dbox{A_i} & 
 \rLine                      
  &\dbox{A_1}&&&\dDerupdto~{u'_n}
\\
\dUpdto<{u_i}&  
 &&
 {\qquad\vdots} &&                                         
\\
\dbox{A'_i}&&\dDerupdto~{u'_1} &&& A''_n
\\ 
\dDerupdto<{\amdx{u_i}} &&&&
\ruDashline(5,3)>{\quad R''}
\ruDashline(3,3)~{..j\noteq i..}  
&
\\ 
&&&\fakeob&&
\\ 
A''_i&
\rDashline
& A''_1&&&
\end{diagram}%
\end{tabular}
\caption{Update propagation pattern 
(To distinguish given data from those produced by the operation, the former are shown with framed nodes and solid lines, 
while the latter are non-framed and dashed.)
}
\label{fig:ppgPatterns-reflect}
\end{figure}

}{
\cellW=7.25ex%
\cellH=5.0ex
\begin{wrapfigure}{R}{0.55\textwidth}
\vspace{-3ex}
	\centering
\begin{tabular}{c} 
\begin{diagram}[w=1.\cellW,h=1.\cellH] 
&& &&&\dbox{A_n}
\\  
&&&&
\ruLine(5,3)>{\qquad R}
\ruLine(3,3)~{..j\noteq i..}
& 
\\ 
&&&\fakeob&&
\\  
\dbox{A_i} & 
 \rLine                      
  &\dbox{A_1}&&&\dDerupdto~{u'_n}
\\
\dUpdto<{u_i}&  
 &&
 {\qquad\vdots} &&                                         
\\
\dbox{A'_i}&&\dDerupdto~{u'_1} &&& A''_n
\\ 
\dDerupdto<{\amdx{u_i}} &&&&
\ruDashline(5,3)>{\quad R''}
\ruDashline(3,3)~{..j\noteq i..}  
&
\\ 
&&&\fakeob&&
\\ 
A''_i&
\rDashline
& A''_1&&&
\end{diagram}%
\end{tabular}
\caption{Update propagation pattern 
(To distinguish given data from those produced by the operation, the former are shown with framed nodes and solid lines, 
while the latter are non-framed and dashed.)
}
\label{fig:ppgPatterns-reflect}
\vspace{-3ex}
\end{wrapfigure}

}
 A general schema of update propagation including reflection is shown in \diarefii. We begin with a consistent multimodel $(\ai...\xai{n},R)$%
 \footnote{Here we first abbreviate $(\ai,\ldots,\xai{n})$ by $(\ai...\xai{n})$, and then write $(\ai...\xai{n},R)$ for $((\ai...\xai{n}),R)$. We will apply this style  in other similar cases, and write, \eg, $i= 1...n$  for $i\in\{1,...,n\}$.
 }
 one of which members is updated \frar{u_i}{A_i}{A'_i}. 
 The propagation operation, based on a priori defined propagation policies as sketched above, produces:
 
\noindent a) updates on all other models \frar{u'_j}{A_j}{A''_j}, $1\le j\noteq i \le n$; 
 
\noindent b) an amendment \frar{\amdx{u_i}}{A'_i}{A''_i} to the original update; 
 
\noindent 
c) a new correspondence specification $R''$ 
such that the updated multimodel 

$\qquad(A''_1...A''_n, R'')$ 

\noindent is consistent. 
 
Below we introduce an algebraic model encompassing several operations and algebraic laws formally modelling situations considered so far. 

  
 It is easy to understand that consistency restoration will need \syncon\ pattern in \diarefii\ rather than \diarefi\ when global consistency constraints imply different granularity of changes on the different  sides. E.g., ``small'' change $u_1$ on the $A_1$ side forces one or more``big'' changes $u'_i$ on the sides of the multimodel, which have to be matched by an additional change $u'_1$ on the $A_1$-side. This completion idea can be formally expressed by the condition that if the original update would be equal to the composition \frar{u_1;u'_1}{A_1}{A''_1}, its propagations $u'_j$  to other models would be the same. 
 It is also clear that the patterns in \diarefi\ is a special cases of pattern in \diarefii\ for which the reflection $u'_i$ is the idle (identity) update.

 An important feature of update propagation scenarios above is that consistency could be restored without changing the model whose update caused inconsistency. However, this is not always possible. Consider an update \frar{u_1}{A_1}{A'_1} that consists of adding a new person $p$ to set $A_1(\person)$.  Suppose that address ${\rm x@gm.com}=p.A_1(\ident)$ is listed in $A_2(\addr)$ as the login address of some member $m$ of a network $n\in A_2(\netw)$. Hence, we add the link $p=m$ to span $r'$. Now, as network membership data are optional for persons in $A_1$, they can be omitted in the updated state $A'_1$. However, if network $n$ belongs to $A_1(\netw)$ as well (formally, there is a link $(n1,n)\in r$ for some $n1\in A_1(\netw)$), in order to respect the global constraint $C_{\memb}$, we need to extend model $A'_1$ with the respective link $(n1,p)\in A'_1(\memb)$).  Moreover, if person $m$ is a member of several other networks in the set $A_1(\netw)\cap A_2(\netw)$, all those membership links must also be added to model $A'_1$. Moreover, login addresses for those networks must be added to the set $p1.A'_1(\emails)$ to respect the global constraint $C_\in$. Thus, consistency restoration requires not only propagating $u_1$ to an update on $A_2$, but also an additional {\em reflective} update \frar{u'_1}{A'_1}{A''_1} as shown in diagram \diarefii\ with $i=1$. 

%

Consider the UML class diagrams in Fig.\ref{fig:example}. They show 3 class diagrams. The intention is to keep track of possible car sharing opportunities as follows: $M_1$ may be a model of an authority to specify employment relations of persons. For each person its address is maintained. Several persons may live at the same address, persons may be unemployed, otherwise relation $empl_1$ assigns to a person its set of companies she works for. $M_2$ is a model (probably at a different authority), which keeps track of companies, their resp.\ locations and the set of their employees. Each person must be an employee, several companies may be located at the same address. 

Finally, $M_3$ specifies a (computed) view of data typed over $M_1$ and $M_2$: For each person $p$ living at $a_1$ and being employed at company $c$, which is located at $a_2$, there is an instance of type ''Home$\to$Work'' ($HW$), which logs the fact that there is someone driving from $a_1$ to $a_2$ such that it is possible to recommend joining $p$ to persons with the same journey to the company. Likewise and redundantly, relation ''Work$\to$Home($WH$)'' is computed to recommend the joining someone returning from the company. Note that we do not require the latter relation to be the inverse of the former in general: It may be possible that there is an $HW$-pair $(a_1, a_2)$ without an $WH$-pair $(a_2,a_1)$. {\em We only claim that $HW = WH^{-1}$ after the above described view computation has been carried out.}

Let's consider data $A_1, A_2, A_3$ typed over $M_1$, $M_2$, $M_3$, resp.: Then there due to the described view computation above, there is an obvious consistency requirement for this data: 
\begin{enumerate}
  \item\label{constr:c1} We define constraint \namefont{viewCons} by requiring formally that $A_3$ is computed from $A_1$ and $A_2$:   
  \[\namefont{viewCons}: \;HW = \{(p.livesAt, c.locatedAt) \mid (p,c)\in empl_1\mbox{ and }(p,c)\in empl_2\}\]
  and 
  \[\;WH = \{(c.locatedAt,p.livesAt) \mid (p,c)\in empl_1\mbox{ and }(p,c)\in empl_2\} \]
\end{enumerate}         
But the authorities also agree on two more requirements (we write $A_i(T)$ for all data elements in $A_i$ being typed over class $T$, e.g. $A_1(Company)$ is the set of all companies in $A_1$): 
\begin{enumerate}
  \item[2]\label{constr:c2} Since $M_2$ is supposed to maintain all known companies in the area, each company being contained in $A_1$ shall also be contained in $A_2$: 
  \[\namefont{compCons}: \;A_1(Company) \subseteq A_2(Company)\]
\item[3]\label{constr:c3} The employment relation shall be consistent! I.e.
\[\begin{array}{c}
\namefont{emplCons}: \; (p,c) \in empl_2 \mbox{ and }p\in A_1(Person) \mbox{ implies }  (p,c)\in empl_1 \\
\mbox{ and } \\
(p,c) \in empl_1 \mbox{ implies } (p,c)\in empl_2 \\
\end{array}
\]
The question now arises, how the collection $A_1, A_2, A_3$ of data shall be synchronised, if an update in one of them occurs. One could, \eg, add or delete persons to / from $M_1$ or $M_2$, add / delete companies to / from $M_1$ or $M_2$ or even add or delete information about commuting to / from companies from / to home. It may also be possible, to change $A_3$ (probably becaise some further information about commuting persons has been tracked which could not directly be derived from data $A_1$ and $A_2$).

We consider several scenarios, but, for the sake of simplicity, we narrow the set of possible updates by considering only atomic additions (deletions are similar): We can add single elements: new persons / companies or employment relations in $M_1$ or $M_2$, new commuting relations in $M_3$. We do not consider adding addresses alone to any of the models, since, in this scenario, we assume addresses to belong to either a person or a company. 

This yields the following {\em Update Propagation Policies}: 
\begin{itemize}
  \item Adding a new (still unemployed) person to $A_1$ or adding a new company (initially without employees) to $A_2$ requires no change in the other data.
  \item Adding a new company $comp$ to $A_1$ possibly requires an update w.r.t. constraint \namefont{compCons}: If there is no $c'\in A_1(Company)$ with $c'.cname = c.cname$, add $c$ to $A_1$.
  \item Adding a new person $p$ to $A_2$ requires no change. 
  \item Adding a new employment relation to either $A_1$ or $A_2$ requires adding it to $A_2$ / $A_1$, resp. and requires to recompute the view $A_3$.
  \item Adding a new element $hw$ of type $HW$ to $A_3$ (we only consider $HW$, updates of $WH$ are similar). We have to distinguish 4 cases: 
  \begin{itemize}
    \item $hw.from\in A_1(Address)$ and $hw.to\not\in A_2(Address)$: This is only possible if all persons living at $hw.from$ are unemployed (we denote this person set with $livesAt^{-1}(hw.from)$). In this case one has to determine from further analysis (there is further information available, \eg, the new address $hw.to$ from which the choice maybe narrowed, but probably manual activities have to be taken into account) which of these persons became employed. For future reference, let's call the operation of finding this person {\em newEmployment}. It assigns to a nonempty person set $P$ an employment relation $(p,c)\in empl_1$ with $p\in P$. Now we have to add $(p,c)\eqdef newEmployment(livesAt^{-1}(hw.from))$ to $A_1$ and to $A_2$ together with company $c$ in $A_2$, if it was not already present, such that $c.locatedAt = hw.to$. Moreover, $A_3$ has to be recomputed.   
    \item $hw.from\not\in A_1(Address)$ and $hw.to\in A_2(Address)$:  Then the set $hw.to.locatedAt^{-1}.employs$, i.e.\ all persons $p$ working for companies located at $hw.to$ must not intersect $A_1(Person)$, because otherwise, by \namefont{emplCons}, if such a $p\in A_1(Person)$, then $(p,c)\in empl_1$ for all $c\in hw.to.locatedAt^{-1}$, which means that $hw$ was already in $A_3$ by \namefont{viewCons}. Hence, we have to update as follows: add $hw.from$ to $A_1$ together with a new person $p$ with unknown (still to be determined) $p.pname$ and $p.livesAt = hw.from$. Then apply $newEmployment(\{p\})$ (see above) and update $A_2$ accordingly.
    \item  $hw.from\not\in A_1(Address)$ and $hw.to\not\in A_2(Address)$: Add $hw.from$ to $A_1$, again together with a new person $p$ with name still to be determined, and add $newEmployment(\{p\})$ to $A_1$ and $A_2$. 
    \item  $hw.from\in A_1(Address)$ and $hw.to\in A_2(Address)$, Then all $p\in hw.from.livesAt^{-1}$ must be unemployed, otherwise, by \namefont{viewCons}, $hw$ was already in $A_3$. Thus we can again determine $newEmployment(hw.from.livesAt^{-1})$ and update $A_1$ and $A_2$ accordingly.
  \end{itemize}
{\em Furthermore, in all four cases, an additional update of $A_3$ is needed to keep everything consistent, namely adding $(hw.to, hw.from)$ to $WH$. Later we will call these updates ''reflective''.}    
\end{itemize}   
\end{enumerate}         
{\color{red}Remark for later use: This is a ternary lens $\ell$. It is wb (w.r.t. Stabilty and Reflect 1/2), can not check weak invertibility!

Star composition can well be explained if we assume that $A_1$ and $A_2$ are XML-Files that have to be synchronised with relational database contents $B_1$, $B_2$ storing this data, as well. $B_3 = A_3$ is an xls-sheet.

An asymmetric lens arises from taking the colimit over the colimited model $M^+$ (can informally be described): Get is projection, put is update propagation according to the just described policies. They satisfy PutGetPut. They can be combined to recover the ternary lens $\ell$. It remains to discuss preservation of wb-properties in this example.}  

\section{
Multimodel spaces
}\label{sec:mmod}

In this section we begin to build a formal \fwk\ for delta lenses: model spaces are described as categories whose objects are models and arrows are updates, which carry several additional relations and operations. We also abstractly define correspondences between models and our central notion of a (consistent) multimodel. We will follow an established terminological tradition (in the lens community) to give, first, a name to an algebra without any equational requirement, and then call an algebra satisfying certain equations {\em well-behaved}.   
\zdmoddok{A delta-based mathematical model for bx  is well-known under the name of delta lenses; below we will say just {\em lens}. There are two main variants: asymmetric lenses, when one model is a view of the other and hence does not have any private information, and symmetric lenses, when both sides have their private data not visible on the other side \cite{me-jot11,me-models11,bpierce-popl11}. In the next sections we will develop a \fwk\ for generalizing the idea for any $n\ge 2$ and including reflective updates.}{}{goes to \sectref{sec:lego}??}{-2em}


\subsection{Background: Graphs, (co)Spans, and Categories}

\label{sec:notation} 
We reproduce well-known definitions to fix our notation.
 A {\em (directed multi-)graph} $G$ consists of a set \xob{G} of {\em nodes} and a set \xarr{G} of {\em arrows} equipped with two functions \frar{\so,\ta}{\xarr{G}}{\xob{G}} 
 that give arrow $a$ its {\em source} $\so(a)$ and {\em target} $\ta(a)$ nodes. We write \frar{a}{N}{N'} if $\so(a)=N$ and $\ta(a)=N'$, and \frar{a}{N}{\_} or \frar{a}{\_}{N'} if only one of these conditions is given. 

 {\renewcommand\xarr[1]{#1}
 Expressions $\xarr{G}(N,N')$, $\xarr{G}(N,\_)$, $\xarr{G}(\_,N')$ denote sets of, resp.,  all arrows from $N$ to $N'$, all arrows from $N$, and all arrows into $N'$. 
}
 
 A pair of arrows \frar{a_i}{N}{N'_i}, $i=1,2$, with a common source is called a {\em [binary] span}  with node $N$ its {\em head} or {\em apex}, nodes $N'_i$ {\em feet}, and arrows $a_i$ {\em legs}. 
\hkmodno{Dually, a pair of arrows \flar{a_i}{N}{N'_i}, $i=1,2$, with a common target is called a {\em (binary) cospan} with apex, feet, and legs, 
defined similarly.}{}{Not needed!} 

A {\em [small] category} is a graph, whose nodes are called {\em objects}, arrows are associatively composable,  and every object has a special {\em identity} loop, which is the unit of the composition. In more detail, given two consecutive arrows \frar{a_1}{\_}{N} and \frar{a_2}{N}{\_}, we denote the composed arrow by $a_1;a_2$. The identity loop of node $N$ is denoted by $\id_N$, and equations $a_1;\id_N=a_1$ and $\id_N;a_2=a_2$ are to hold. 
We will denote categories by bold letters, say, \spA, and often write $A\in{\spA}$ rather than $A\in\xob{\spA}$ for its objects. 

A {\em functor} is a mapping of nodes and arrows from one category to another, which respects sources and targets as well as identities and composition. Having a tuple/family of categories $\spA=(\spA_1...\spA_n)$, their {\em product} is a category $\prod\!\spA=\spA_1\timm...\timm\spA_n$ whose objects are tuples $\tupA=(A_1...A_n)\in\xob{\spaA_1}\timm...\timm\xob{\spaA_n}$, and arrows from $(A_1...A_n)$ to $(A'_1...A'_n)$ are tuples of arrows $\tupu=(u_1...u_n)$ with \frar{u_i}{A_i}{A'_i} for all $i=1...n$.   

\hkmodok{Given a category \spC, we can build another category $\Span(\spC)$ as follows\ldots}{}{Would be too much notation}

\facend

\subsection{Model Spaces and updates}
\label{sec:modelspaces}
Basically, a {\em model space} is a category, whose nodes are called {\em model states} or just {\em models}, and arrows are {\em (directed) deltas} or {\em updates}. For an arrow \frar{u}{A}{A'}, we treat $A$ as the state of the model before update $u$, $A'$ as the state after the update, and $u$ as an update specification.  Structurally, it is a specification of correspondences between $A$ and $A'$. Operationally, it is an edit sequence (edit log) that changed $A$ to $A'$. The formalism does not prescribe what updates are, but assumes that they form a category, \ie, there may be different updates from state $A$ to state $A'$; updates are composable; and idle updates \frar{\id_A}{A}{A} (doing nothing) are the units of the composition. 
\hkmodok{
It is also useful to have a more concrete notion of a model space as based on a category \spE\ of models and model embeddings, and then define an update \frar{u}{A}{A'} as a span of embeddings \flar{u_l}{A}{K}  and \frar{u_r}{K}{A'} where $K$ is a submodel of $A$ consisting of all elements preserved by $u$ so that embeddings $u_l$ and $u_r$ specify the deleted and the inserted elements resp. Then a model space \spA\ appears as the category $\Span(\spE)$ for some given \spE. Two ``extreme'' cases for the latter are a) the category of graphs and their inclusions, $\gracat_\subset$, which does not distinguish between objects/references and primitive values/attributes, and b) the Object-Slot-Value model developed in \cite[Sect.3]{me-gttse09} , which does make the distinction between objects and values but does not consider references/links between objects. Integrating a) and b) gives us the category of attributed graphs \cite{gtbook06}. Some of the notions that we axiomatically define below can be given formal definitions for a concrete choice of the category \spE, and we refer to paper \cite{orejas-bx13}, in which some of these definitions can be found for $\spE=\gracat_\subset$.
}{A prominent example of model spaces is the category of graphs where updates are encoded as (certain equivalence classes of) binary spans between them. They are heavily used in the theory of Graph Transformations \cite{gtbook06}. In this way an update $u:A\to A'$ can be a deletion or an addition or a combination of both.}{Thus updates as spans are at least mentioned informally. The formal treatment of the span cat with equivalences and so on is too much here! }
\par
We require every model space $\spA$ to be endowed with two additional constructs of update compatibility. 


\subsubsection{Sequential compatibility of updates} 
We assume a family $(\Kupd_A)_{A\inn\xob{\spaA}}$ of binary relations 
{\renewcommand\xarr[1]{{#1}}
$\Kupd_{A}\subset\xarr{\spaA}(\_,A)\timm\xarr{\spaA}(A,\_) $
}
indexed by objects of \spaA, and 
specifying {\em non-conflicting} or {\em compatible} consecutive updates. Intuitively, an update 
$u$ into $A$ is compatible with update 
$u'$ from $A$, if $u'$ does not revert/undo anything done by $u$, \eg, it does not delete/create objects created/deleted by $u$, or re-modify attributes modified by $u$. For example, one could add Mary's employment at IBM ($\frar{u}{A_1}{A_1'}$) and subsequently add Ann ($\frar{u'}{A_1'}{A_1''}$) to $A'_1$ yielding a pair $(u,u')\in \Kupd_{A_1'}$ (see \cite{orejas-bx13} for a detailed discussion). Later we will specify several formal requirements to the compatibility (see \defref{def:wbspace} below). 


{
\renewcommand\ldots{}
\cellW=7.25ex%
\cellH=5.0ex
\begin{figure} \centering
	\centering
\begin{equation}\label{eq:mergeDef}
\begin{diagram}[w=1.\cellW,h=1.\cellH] 
	&&&\dbox{A}&
	\\  
	&&\ldTo(3,2)^{u_1}\ldots^{\hspace{6ex}--\Kdisj{}--}\hspace{2ex}
	   \ldTo(1,2)_{u_2}&&
	  \\ 
	\dbox{A'_1}&\ldots\hspace{10ex}[\merge] & \dbox{A'_2}&\dDerupdto>{u_1{+}u_2} &
	\\  
	&\rdDerupdto(3,2)_{\widetilde{u_1}}&&\ldots\rdDerupdto(1,2)^{\widetilde{u_2}}&
		\\ 
	&&&A''&
	\\
\end{diagram}%
\end{equation}
\caption{Update merging}\label{fig:updMerge}
\end{figure}
}

\subsubsection{Concurrent compatibility of updates and their merging} 
Intuitively, a pair of updates $u=(u_1, u_2)$ from a common source $A$ (\ie, a span) as shown in \figref{fig:updMerge} diagram \eqref{eq:mergeDef} is called {\em concurrently compatible}, if it can be performed in either order leading to the same result -- an update \frar{u_1+u_2}{A}{A''}. Formally, in the case of concurrent compatibility of $u_1$ and $u_2$, we require the existence of update \frar{u_1+ u_2}{A}{A''} 
and updates $\frar{\tilde{u_i}}{A_i'}{A''}$ such that $u_1;\tilde{u_1} = u_2;\tilde{u_2}=u_1+u_2$. Then we call updates $u_i$ {\em mergeable}, \zdmodok{complement updates $\tilde{u_1},\tilde{u_2}$}{update $u_1+u_2$}{}  their {\em merge}, and
updates $\tilde{u_1},\tilde{u_2}$ {\em complements}, and write $(\tilde{u_1},\tilde{u_2})=\merge_A(u_1,u_2)$ or else $\tilde{u_i}=\merge_{A,i}(u_1,u_2)$. We will also denote the model $A''$ by $A'_1+_A A'_2$. For example, for model $A_1$ in Fig.\ref{fig:multimod-uml}, we can concurrently delete John's and add Mary's employments with IBM, or concurrently add two Mary's employments, say, with IBM and Google. But deleting Mary from the model and adding her employment with IBM are not concurrently compatible. Similarly, in $A_3$, updating addresses of different routes, or updating the \from\ and \toattr\ attributes of the same route are concurrently compatible, but deleting a route and changing its attributes are incompatible (we will also say, {\em in conflict}). 
%
%
We denote the set of all mergeable spans with apex $A$ by $\Kdisj{}_A$.

The definition of concurrent compatibility is a generalization of the notion of parallel independence of graph transformation rules \cite{gt-EEPT-2006}. Below we will elaborate further on the interplay of sequentially and concurrently compatible pairs.

Now we can define model spaces.
\begin{defin}[Model Spaces] \label{def:mspace}
	A {\em model space} is a tuple 
	\[\spaA=(|\spaA|, \Kupd, \Kdisj{}, \merge)\]
	of the following four components. The first one is a category |\spaA| (the {\em carrier}) of {\em models} and {\em updates}. We adopt a notational convention to omit the bars $|$ and denote a space and its carrier category by the same symbol \spA. 
	The second component is a family 
	\[\Kupd=\comprfam{\Kupd_A\subset \spA(\_,A)\timm\spaA(A,\_)}{A\inn\xob{\spA}}\]	
of {\em sequential compatibility} relations  for sequential update pairs as described above. The third and forth component are tightly coupled:
	\[\Kdisj{}=\comprfam{\Kdisj{}_A\subset \spA(A, \_)\timm\spA(A,\_)}{A\inn\xob{\spA}}\]
is a family of {\em concurrently compatible} or {\em mergeable} spans of updates, and 
\[\merge =\comprfam{ \merge_A}
	{A\inn\xob{\spA}}\]
is a family of merge operations as shown in \eqref{eq:mergeDef}. The domain of operation $\merge_A$ is exactly the set $\Kdisj{}_A$. 
We will denote the components of merge by $\merge_{A,i}$ with $i=1,2$ and omit index $A$ if they are clear from the context or not important. Writing $\merge_A(u_1, u_2)$ implicitly assumes $(u_1, u_2) \in \Kdisj{}_A$. \qed
\end{defin}
\begin{defin}[Well-behaved Model Spaces]\label{def:wbspace}
	\noindent 
{\renewcommand\&{\wedge}
\begin{center}
\begin{tabular}{l@{\quad}l}
	\lawnamebr{\Kupd\ and~ Id}{} &	
	For all $u\inn{\spA}(\_,A)$, $u'\inn{\spaA}(A,\_)$: $(u,\id_A), (\id_A,u')\inn\Kupd_A$, 
	\\ [1.5ex]
		&	
For any three consecutive updates $u,v,w$, we require: 
	\\  
		\lawnamebr{\Kupd\Kupd }{1}&
			$(u\Kupd\,vw$ $\&$ $v\Kupd w)$ imply 
			$(u\Kupd v$ $\&$ $uv\, \Kupd w)$,
			\\
		\lawnamebr{\Kupd\Kupd }{2} &
		 $(uv\,\Kupd w$ $\&$ $u\Kupd v)$ imply $(v\Kupd w$ $\&$ $u\Kupd\,vw)$,
			 \\
			 & where composition is denoted by concatenation ($uv$ for $u;v$ etc). 
	\\ [1ex]
\lawnamebr{\Kdisj{}~and ~\Kupd}
{} &
For all $(u_1,u_2)\in\Kdisj{}_A$: $(u_i, \widetilde{u_i})\in\Kupd_{A'_i}$, $i=1, 2$,
where $\widetilde{u_i}=\merge_{A,i}(u_1, u_2)$
		\\ [1.5ex]
	\multirow{2}{*}{\lawnamebr{MergeSym}{}}& For all $\frar{u_1}{A}{A'_1}$, $\frar{u_2}{A}{A'_2}$, if $(\widetilde{u_1}, \widetilde{u_2})=\merge_A(u_1, u_2)$,
	\\ & then $(\widetilde{u_2}, \widetilde{u_1})=\merge_A(u_2, u_1)$ 
	\\ [1.5ex]
	\lawnamebr{MergeId}{}& For all $\frar{u}{A}{A'}$: $\merge_A(u,\id_A) = (\id_{A'},u)$    
\end{tabular}
\end{center} \qed 
}  
\end{defin}
The first condition is the already discussed natural property for sequential compatibility. The pair of conditions below it requires complementing updates not to revert anything done by the other update. The last three conditions are obvious requirements to the operation \merge. 
\par
We assume all our model spaces to be wb and, as a rule, will omit explicit mentioning. Each of the metamodels $M_1$..$M_3$ in the running example gives rise to a (wb) model space of its instances as discussed above.

\facend
\subsection{Correspondences and Multimodels}
\zdmodtok{This is the newer non-functorial version of the section. To switch to the functorial one -- switch from multimodSpaces-last to multimodSpaces-new}{}{}
We will work with families of model spaces indexed by a finite set $I$, whose elements can be seen as space {\em names}. To simplify notation, we will assume that $I=\{1,\ldots,n\}$ \zdmoddno{although ordering will not play any role in our \fwk.}{}{Importantly!!}{-2ex} Given a tuple of model spaces $\spA_1, \ldots, \spA_n$, we will refer to objects and arrows of the product category $\prod\spA=\spaA_1\timm\cdots\timm\spaA_n$ as {\em tuple models}  and {\em tuple updates},  
 and denote them by letters without indexes, \eg, \frar{\tupu}{\tupA}{\tupA'} is a family \frar{u_i}{A_i}{A'_i}, $i=1..n$.  
We will call components of tuple models and updates their {\em feet}. We also call elements of a particular space $\spA_i$ {\em foot} models and {\em foot} updates. 

\zdmoddtok{For this paper, we are interested in a special class of tuple updates, whose all feet besides one are identities, \ie, updates of the type $(\id_{A_1}..u_i..\id_{A_n})$ for some $i\le n$; we will call such updates {\em singular}. 
Given a tuple model $A=(A_1...A_n)$, any foot update \frar{u_i}{A_i}{A'_i} gives rise to a unique singular update \frar{\ovr{u_i}}{A}{\ovr{A'_i}}, 
where $\ovr{A'_i}=(A_1..A'_i..A_n)$ and $\ovr{u_i}=(\id_{A_1}..u_i..\id_{A_n})$ (and any singular update is such). Note also that any tuple update $u=u_1...u_n$ can be represented as a sequential composition $\ovr{u_1}...\ovr{u_n}$ of singular updates (taken in any order). 
}{}{not sure but i seems we don;t need it}{-4em}

\begin{defin}[Multispaces, Alignment, Consistency]
\label{def:mac}
	Let $n\ge 2$ be a natural number. An {\em n-ary multi-space \zdmodtok{(with alignment)}{}{}}  is a triple $\mspA=(\spA, \Corr, \prt)$ with the following components. The first one is a tuple of (wb) model spaces $\spA = (\spaA_1, \ldots, \spaA_n)$ called the {\em boundary} of \mspA. 
	The other two components is a class \Corr\ of elements called {\em (consistent) correspondences} or {\em corrs}, and a family of {\em boundary} mappings $\prt=\comprfam{\frar{\prt_i}{\Corr}{\xob{\spA_i}}}{i=1..n}$.  
	
	\newcommand\prtbold{\ensuremath{\boldsymbol{\prt}}}
	A corr $R\in\Corr$ is understood as a correspondence specification interrelating models $A_i=\prt_iR$; the latter are also called {\em R's feet}, and we say that models $A_i$ are aligned via $R$. We write $\prt R$ for the tuple $(\prt_1R...\prt_nR)$. When we consider several multispaces and need an explicit reference, we will write \prtbold \mspA\ for the boundary \spA\ of the entire multispace, and \Corrx{\mspA} for its class of corrs. 

\newpii{
Given a model tuple  $A=(A_1...A_n)\in\xob{(\prod\spA)}$, we write  
	$\Corrx{A}$ for the set \compr{R\in\Corr}{\prt R=A%
	}, 
and $\Corrx{A_i}$ for the class  \compr{R\in\Corr}{\prt_i R=A_i}; 
thus, $\Corrx{A}=\bigcap_{i=1..n}\Corrx{A_i}$.}{1}{-3em}
\footnote{With this line of notation, the entire class \Corrsma\ could be  denoted by \Corrx{\spA}.} 
	
Although in this paper all corrs are considered consistent by default, we will often mention their consistency explicitly to recall the context and to ease comparison with similar but a bit more general \fwk s, in which inconsistent corrs are also considered \cite{me-models11,me-jss15}. \qed

\end{defin}
\zdmodtok{the new def above is in the section "lensesFinal/multimodSpaces-last" while the previuous functorial def is in section ".../multimodSpaces" or maybe ".../multimodSpaces-new".}{}{}
\begin{defin}[Multimodels]\label{def:mmodel}	
  A {\em (consistent) multimodel} over a multispace \mspA\ 
  is a couple $\nia=(A, R)$ of a model tuple $\tupA=(A_1..A_n)$\ with a (consistent) corr $R\in\Corrx{\tupA}$ relating the component models. 
\newpii{A {\em multimodel update} \frar{\mmu}{\nia}{\nia'} is a pair $(u, (R,R'))$, whose first component is a tuple update 
$u=\comprfam{\frar{u_i}{\prt_iR}{\prt_iR'}}{i=1..n}$, and the second component is a pair of the old and the new corrs. Identity updates are pairs $(\id, (R,R))$, whose tuple component consists of identities $\id_{\prt_iR}$, $i=1..n$, only. It is easy to check that so defined multimodels and their updates determine a category that we denote by \mspR. 
}{2}{-6.5em}
	\qed
\end{defin}
\begin{remark}\label{rem:corr-upd}\newpii{
	The notions of a corr and a multimodel are actually synonyms: any corr $R$ is simultaneously a multimodel $(\prt R,R)$, and any multimodel $(A,R)$ is basically just a corr $R$ as the feet part of the notion is uniquely restored by setting $A=\prt R$, thus, $\Corr\cong\xob{\mspR}$. We can extend the equivalence to arrows too by defining corr updates exactly as we defined multimodel updates. This would make class \Corr\ into a  category  $\Corrcat$ isomorphic to $\mspR$. Using two symbols and two names for basically the same notion is, perhaps, confusing but we decided to keep them to keep track of the historical use of the terminology. The choice of the word depends on the context: if we focus on the $R$ component of pair $(A,R)$, we say ``corr'', if we focus on the $A$ component, we say ``multimodel''.
}{3: the Remark is new}{-5em}

\newpii{
	Thus, for this paper, multimodel updates are basically tuple updates $u$ while their corr-component $(R,R')$ is trivial (co-discrete in the categorical jargon) and so we actually will not need the category \Corrcat\ (or \mspR) explicitly declared. However, the notion will be useful when we will discuss categorification of the \fwk\ in Future Work \sectref{sec:future}.1, and in Related Work \sectref{sec:related}. 
}{3 cont'd}{-6em}
\end{remark}

%

\begin{example}
The \exIntro\ of Sect.\ref{sec:example} 
gives rise to a 3-ary multimodel space. For $i=1..3$, space $\spA_i$ consists of all models instantiating metamodel $M_i$ in  Fig.\ref{fig:multiMetamod-uml} and their updates. Given a model tuple $A=(A_1, A_2, A_3)$, a consistent corr $R\in \Corrx{A}$ is given by a triple of relations $\ptop^R$, $\ctoc^R$, and $\etor^R$ (the first is specified by formula \riii\ on p.~\pageref{pg:constr123} and the other two by \rii\ on p.~\pageref{pg:constr12}) such that the intermodel constraints (C1-3) are satisfied.
If we rely on person and company names as keys, then relations $\ptop^R$, $\ctoc^R$ are derived from foot models $A_1$ and $A_2$, but if these keys are not reliable, then the two relations are independent components of the multimodel. Relation $\etor^R$ is always an independent component. 

\zdmoddtok{
Realignment (for foot, and hence all) updates is defined as follows. Derivable relations are simply recomputed for the target model tuple after a local update $u_i$. The case of non-derivable relations is more interesting. Suppose, for example, that \ptop\ is an independent relation, and a new person was added to, say, model $A_1$. Nothing changes in the correspondences as deciding whether this new person in $A'_1$ is identical to an existing person in $A_2$  cannot be done automatically in the absence of keys. If a person is deleted from $A_1$, and this person appeared in $\ptop^R$, then the \corring\ pair is to be deleted from $\ptop^R$ yielding an updated relation $\ptop'$.   
In a more refined setting, relation $\ptop^R$ can be encoded as a binary span between $A_1$ and $A_2$, and an update \frar{u_1}{A_1}{A'_1} is also a binary span (relation), cf.\ the introductory remarks in Sect.~\ref{sec:modelspaces}. Then realignment amounts to relation composition (denoted  by $\Join$): $\ptop'=\ptop^R\Join \restr{u_1}{\Person}$ where \restr{u_1}{\Person} is the part of the update affecting set $\Person^{A_1}$.
Two other correspondence relations, \ctoc\ and \etor, can be processed similarly, \eg, $\etor'=\etor^R\Join \restr{u_1}{\Empl}$. 
Thus, we \hkmodok{obtain}{can automatically compute}{important!} the entire new correspondence $R'$ by defining  $\ptop^{R'}=\ptop'$, $\ctoc^{R'}=\ctoc'$, and $\etor^{R'}=\etor'$. Finally, we set $u_1^\bigstar(R)=R'$. Functoriality of $\_^\bigstar$ follows from the associativity of relational composition: the idea is given by the following schema 
	//
$(u_1;v_1)^\bigstar(R)=(u_1;v_1)\Join R= (u_1\Join v_1)\Join R=u_1\Join(v_1\Join R)=u_1^\bigstar(v_1^\bigstar(R))$. 
//
An accurate proof needs an accurate definition of all components, see \cite{DBLP:conf/ecmdafa/KonigD17} 
}{}{ALIGN}{-20em} 

\par
\hkmodok{
Consistency is determined by intermodel constraints (C1-3) considered on p.~\pageref{pg:constr}, and the multimodel shown in \figref{fig:multimod-uml} is inconsistent. To make it consistent, 
we can add, \eg,  to \ai\ an \eer-link connecting Mary to IBM, add to \aiii\ a commute with attributes $\from=a1$ and $\too=a15$, and relate it to John's and Mary's employment with IBM. 
}{}{This piece seems to be misplaced here. Moreover, it is redundant (with parts of Chapter 2)}
\end{example}

\section{Update Propagation and Multiary (Delta) Lenses}
\label{sec:mlenses}

Update policies described in \sectref{sec:example} can be extended to cover propagation of all updates $u_i$, $i\in 1...3$ according to the pattern  in \figref{fig:ppgPatterns-reflect}. \zdnewwok{This is a non-trivial task, but after it is accomplished, we obtain a \syncon\ \fwk, which we algebraically model as an algebraic structure called a {\em (very) well-behaved lens}. In this term,  {\em lens} refers to a collection of diagram operations defined over a multispace of models, each of which takes a configuration of models, corrs and updates, and returns another configuration of models, corrs and updates --- then we say that the operation {\em propagates} updates. A lens is called {\em well-behaved}, if its propagation operations satisfy a set of algebraic laws specified by equations. This terminological discipline goes back to the first papers, in which lenses were introduced \cite{bpierce-popl05}. We define and discuss well-behaved lenses in the next \sectref{sec:wblens}. Additionally, in \sectref{sec:vwblens}, we discuss yet another important requirement to a reasonable \syncon\ \fwk: compatibility of update propagation with update composition, which is specified by the most controversial amongst the lens laws -- the (in) famous Putput. We define a suitably constrained version of the law and call it Kputput, and call a well-behaved lens satisfying the KPutput law {\em very} well-behaved (again following the terminological tradition of \cite{bpierce-popl05}). }{}{-2em} 

\subsection{Well-behaved lenses}\label{sec:wblens}	

\begin{defin}[Symmetric lenses]\label{def:multi-lenses}
	An $n$-ary {\em symmetric lens} is a pair $\ell=(\mspaA, \ppg)$ with \mspaA\ an $n$-ary multimodel called the {\em carrier} of $\ell$, and  $\ppg=\comprfam{\ppg_i}{i=1..n}$ a family of operations of the following arities. Operation $\ppg_i$ takes a consistent corr $R$ with boundary $\prt R=\tupleA$,
	and a foot update \frar{u_i}{A_i}{A'_i} as its input, and returns three data items (a,b,c) specified below. 
	
	(a) an $(n-1)$-tuple of updates \frar{u'_j}{A_j}{A''_j} with $1\le j\not= i\le n$; 
	
	\newpii{
	(b) an {\em amendment} \frar{\amdx{u_i}}{A'_i}{A''_i} to the original update so that we have a new  {\em reflective} update \frar{u'_i=u_i;\amdx{u_i}}{A_i}{A''_i}
\\
	\hspace{2cm} These data define a tuple update \frar{u'}{A}{A''=(A''_1,..,A''_i,..,A''_n)}. 
}{1}{-2em}	
    
    (c) a new corr $R'\in \Corrx{A''}$.

	\noindent In fact, operation $\ppg_i$ completes a (local) foot update $u_i$ to a (global) update of the entire multimodel \frar{\mathbf{u}}{\mathcal{A}}{\mathcal{A'}}, whose  components are $(u_j')_{j\ne i}$, $u'_i=u_i;\amdx{u_i}$, and the pair $(R,R')$ 
	(see also  \figref{fig:ppgPatterns-reflect}). 
	\qed
\end{defin}
\zdnewok{Note that all \ppg\ operations are only defined for consistent corrs and return consistent corrs. The latter requirement is often formulated as a special lens law (often called \lawname{Correctness}) but in our \fwk,  it is embedded in the arity of propagation operations.}{important} 

%
%
{\bf Notation.} If the first argument \acorr\ of operation \ppg$_i$ is fixed, the corresponding family of unary operations (whose only argument is $u_i$) will be denoted by $\ppgr_i$. By taking the $j$th component of the multi-element result, we obtain single-valued unary operations $\ppgr_{ij}$ producing, resp. updates 
$\frar{u'_j=\ppgr_{ij}(u_i)}{A_j}{A''_j}$
for all $j\not=i$ (see clause (a) of the definition) while $\ppgr_{ii}$ returns the amendment $\amex{u_i}$. 
We also have operation $\ppgr_{i\star}$ returning 
a new consistent corr $R''=\ppgr_{i\star}(u_i)$ according to (c).
\begin{defin}[Closed updates]
	Given a lens $\ell=(\mspaA,\ppg)$ and a corr $R\in \Corr{\tupleA}$, we call an update \frar{u_i}{A_i}{A'_i} $R$-{\em closed}, if $\ppgr_{ii}(u_i)=u_i$ 
	(\ie, $\amdx{u_i}=\id_{A_i'}$). 
	An update $u_i$ is {\em closed} if it is $R$-closed for all $R\in\Corrx{A_i}$. Lens $\ell$ is called {\em closed} 
	at foot $\spaA_i$, if all updates in $\spA_i$ \ 
	are $\ell$-closed. \qed
\end{defin} 
\zdmodtok{there was an idea to write here a piece explaining why writing a full lens for the running example is so laborious and actually not needed. Some pieces are commented in the source }{}{}
%
\begin{defin}[Well-behaved lenses]\label{def:wblens}
	A lens $\ell=(\mspaA,\ppg)$ is called {\em well-behaved (wb)} if the following laws hold for all $i=1..n$, $A_i\in\xob{\spaA_i}$, 
	\zdmodtok{$R\in \Kcorr _{A_i}\subset \Corr{A_i}$}{$R\in  \Corrx{A_i}$}{} and \frar{u_i}{A_i}{A'_i}, cf. \figref{fig:ppgPatterns-reflect}
	
	\lawgap=1ex


\smallskip\noindent
\begin{tabular}{l@{\quad}l}
	\lawnamebr{Stability}{i}&
	$\ppgr_{ij}(\id_{A_i})=\id_{A_j}$\mbox{for all $j=1...n$, and } $\ppg^R_{i\star}(\id_{A_i}) = R$
	\\[\lawgap]
	\lawnamebr{Reflect1}{i} &	
	$(u_i, \amex{u_i})\in \Kupd_{A_i'}$ 
	\\[\lawgap]
	\lawnamebr{Reflect2}{ij} &
	$\ppgr_{ij}(u_i;\amex{u_i}) = \ppgr_{ij} (u_i)$ for all $j\noteq i$
	\\  [\lawgap]  
	\lawnamebr{Reflect3}{i} &
	$\ppgr_{ii}(u_i;\amex{u_i})=\id_{A''_i}$ 
	where $A''_i$ is the target model of $\amex{u_i}$,
\end{tabular}
\\ [\lawgap]
\newpii{where in laws \lawnamebr{Reflect1{-}3}{}, \amex{u_i} stands for $\ppgr_{ii}(u_i)$}{1}{-0.5em}
\end{defin}
\zdmoddtok{\lawname{Hippocr} law says that if realignment returns a consistent corr, nothing should be done.}{\lawname{Stability} says that lenses do nothing voluntarily.}{HIPPO}{-1em}
 \lawname{Reflect1} says that amendment works towards ``completion'' rather than ``undoing'', and  \lawname{Reflect2-3} are idempotency conditions to ensure the completion indeed done. 

\zdmodtok{Stability as Corollary of Hippocr is commented in the source}{}{HIPPO}
%
%

\begin{defin}[{\sf Invertibility}] \label{def:invert}
	A wb lens is called {\em [weakly] invertible}, if it satisfies the following law for any $i$, update \frar{u_i}{A_i}{A'_i} and $R\in \corr{A_i}$:
	
	\noindent \begin{tabular}{l@{\quad}l}
		\lawnamebr{Invert}{i}
		& for all $j\noteq i$: 
		$\ppgr_{ij}(\ppgr_{ji}(\ppgr_{ij}(u_i))) = \ppgr_{ij}(u_i)$	
	\end{tabular}
	\qed 	
\end{defin}
This law deals with ``round-tripping'': operation $\ppgr_{ji}$ applied to 
update $u_j=\ppgr_{ij}(u_i)$ results in update $\hat{u_i}$ equivalent to $u_i$ in the sense that $\ppgr_{ij}(\hat{u_i})=\ppgr_{ij}(u_i)$  (see \cite{me-models11} for a motivating discussion). 
\zdmodno{}{Hmm, \\
	\noindent \begin{tabular}{l@{\quad}l}
		\lawnamebr{Invert{-}strong}{i}
		& for all $j\noteq i$ and all $k$: 
		$\ppgr_{ik}(\ppgr_{ji}(\ppgr_{ij}(u_i))) = \ppgr_{ik}(u_i)$	
	\end{tabular}
}{I just recognized that the much stronger version on the left would perhaps be more reasonable: to think about }
\begin{example}[Trivial lenses \trivlensna]\label{ex:TerminalLens} 
	A category consisting of one object and one (necessary identity) arrow is called {\em terminal}. All terminal categories are isomorphic; we fix one, whose object is denoted by $1$ while the category is denoted by \trmcat\ (bold 1). The terminal category \trmcat\ gives rise to a unique {\em terminal space} with  $\Kupd_1 = \{(id_1,id_1)\}$, $\Kdisj{}_1=\{\id_1, \id_1\}$ and $\merge(\id_1, \id_1)=\id_1$.
	
	Any model space \spA\ gives rise to the following {\em trivial} $n$-ary lens \trivlensna. The first foot space $\spA_1=\spA$ while for all $j=2..n$, $\spA_j=\trmcat$.   
	Tuple models are uniquely determined by their first foot, and given such a model $(A,1...1)$, the set of corrs is the singleton $\{\etosing{A}\}$ with $\prt_1\etosing{A}=A$ and $\prt_{2..n}\etosing{A}=1$, and this only corr is considered consistent. Hence, all mutlimodels are consistent and update propagation is not actually needed. For any \spA, lens $\trivl_n(\spA)$ is a wb, invertible lens closed 
	at all its feet in a trivial way.\qed  
\end{example}

The next example is more interesting.
\begin{example}[Identity Lenses \idlensna]\label{ex:idLens} 
	Let $\spaA$ be an arbitrary model space. It generates an $n$-ary lens \idlensna\ as follows. The carrier $\mspaA$ has $n$ identical feet spaces: $\prtbf_i\mspaA=\spaA$ for all $i=1..n$. The corr set $\Corr{A}$ for a tuple model $ A=(A_1..A_n)$ is the singleton $\{\etosing{A}\}$ with $\prt_i\etosing{A}=A_i$; this corr is consistent iff $A_1=A_2={...}=A_n$. 
	All updates are propagated to themselves (hence the name {\em identity lens}). Obviously, \idlensna\ is a wb, invertible lens closed 
	at all its feet. \qed 
\end{example}

\subsection{Very well-behaved lenses}\label{sec:vwblens}

We consider an important property of update propagation---its compatibility with update composition. A simple compositionality law would require that the composition $u_i;v_i$ of two consecutive foot updates \frar{u_i}{A_i}{A'_i}, \frar{v_i}{A'_i}{B_i}, is propagated into the composition of propagations: 
$$\ppgr_{ij}(u_i;v_i)=\ppgr_{ij}(u_i);\ppgrx{R''}_{ij}(v_i)$$
 with $R''$ being the corr provided by the first propagation, $R''=\ppgr_{i\star}(u_i)$. It is however well known that such a simple law (called PutPut) often does not hold. 

\begin{figure}[t]
\centering
    \includegraphics[
                    width=0.65\textwidth,%
                 ]%
                 {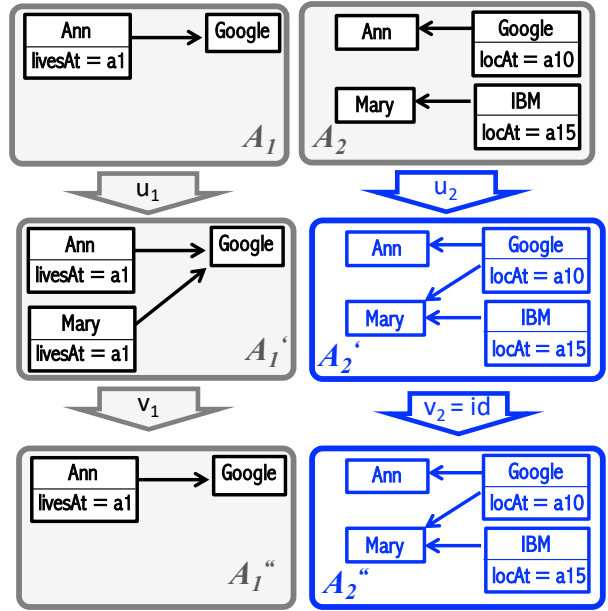}
\caption{Unconstrained Putput is violated - Example 1 \label{fig:ex-putput1}}
\end{figure}
Figure~\ref{fig:ex-putput1} presents a simple example \newpii{(we use a more compact notation, in which values of the attribute \name\ are used as OIDs -- the primary key attribute idea)}{1}{-2em}.  At the initial moment, the binary multimodel $(\ai, \aii)$ with the (implicit) corr given by name matching is consistent: the only intermodel constraint (C1) is satisfied. Update $u_1$ adds a new employment record to model \ai, constraint (C1) is violated, and to restore consistency, a new employment is added to \aii\ by update $u_2=\ppgr_{12}(u_1)$ (the new propagated model and update are shown with blue lines and blank).  Then  update $v_1$ deletes the record added by $u_1$, but as the resulting multimodel $(A''_1, A'_2)$  remains consistent, nothing should be done \newpii{(if we follow the Hippocraticness principle in bx introduced by Stevens \cite{stevens-sosym10})}{2}{-2ex} and thus $v_1$ is propagated to identity, $v_2=\ppgrx{R''}_{12}(v_1) =\id_{A'_2}$. Now we notice that the composition $(u_1;v_1)$ is identity, and hence is to be propagated to identity, \ie, $\ppgr_{12}(u_1;v_1) = \id_{A_2}$, while $\ppgr_{12}(u_1);\ppgrx{R''}_{12}(v_1) =u_2;v_2=u_2 \noteq \id_{A_2}$, and $A''_1\noteq A''_2$.%
\footnote{\label{fnote:oids}
		In more detail, equality of models $A_1$ and $A''_1$ is a bit more complicated than shown in the figure. Objects Ann in \ai\ and Ann in $A''_i$ will actually have different OIDs, but as OIDs are normally invisible, we consider models up to their isomorphism w.r.t. OID permutations that keep attribute values unchanged. Hence, models \ai\ and $A''_i$ are isomorphic and become equal after we factorize models by the equivalence described above.
	}

 Obviously, the violation story will still hold if models \ai\ and \aii\ are much bigger and updates $u_1$ and $v_1$ are parts of bigger updates (but Mary should not appear in \ai).  

\begin{figure}[t]
\centering
    \includegraphics[width=0.9\textwidth]%
                 {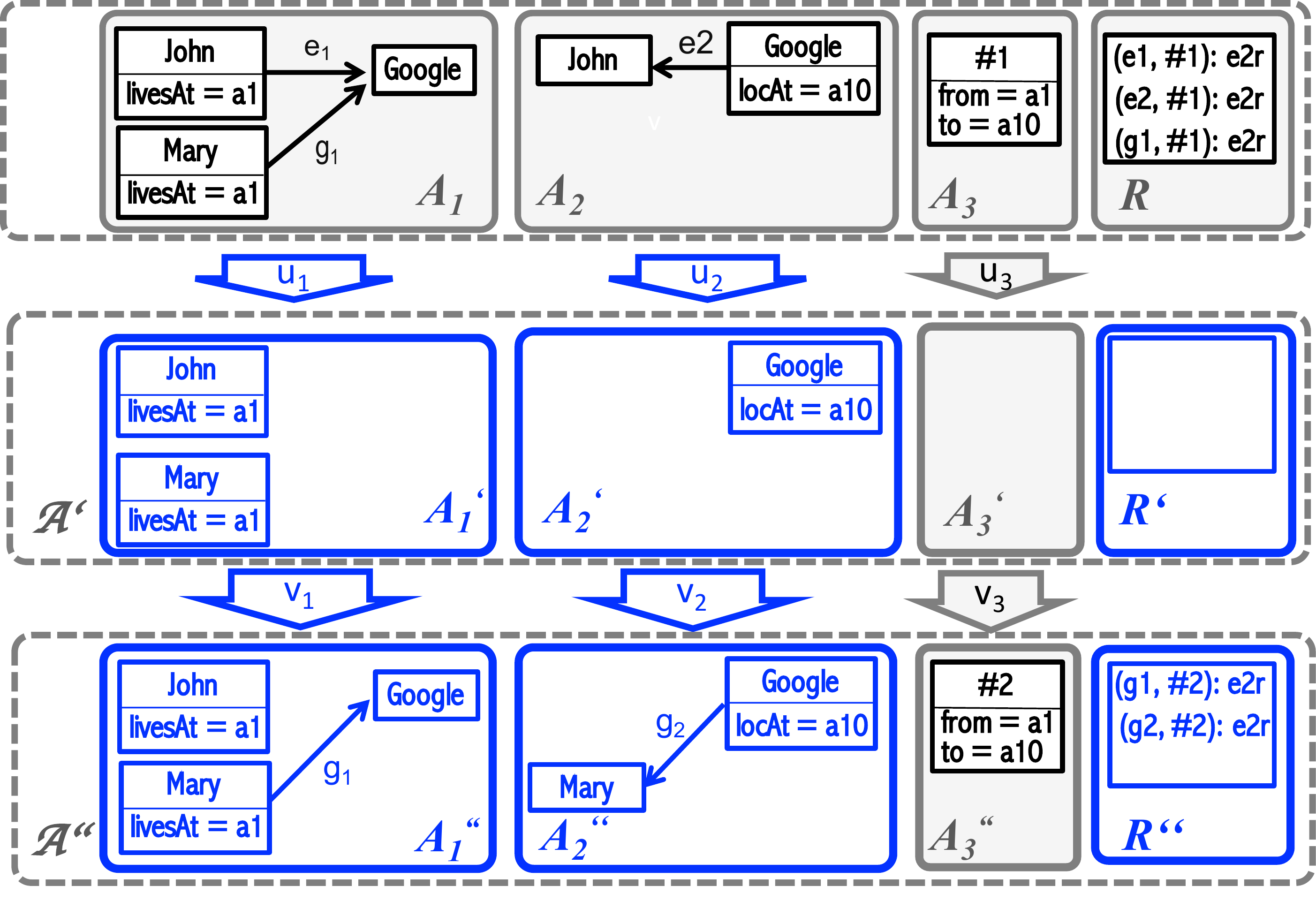}     
\caption{Unconstrained Putput is violated--Example 2 \label{fig:ex-putput2}}
\end{figure}
Figure~\ref{fig:ex-putput2} presents a more interesting ternary example. Multi-model $\nia$ is consistent, but update $u_3$ deletes route \#1 and violates constraint (C2). In this case, the most direct propagation policy would delete all employment records related to \#1\ but, to keep changes minimal, would keep people in $A_1$ and companies in $A_2$ as specified by model $\nia'$ in the figure. Then a new route is inserted by update $v_3$  (note the new OID), which has the same \from\ and \toattr\ attributes as route \#1. There are multiple ways to propagate such an insertion. \newpii{The simplest one is to do nothing as the multimodel $(A'_1, A'_2, A''_3, R'=\varnothing)$ is consistent (recall that mapping \etor\ is not necessarily surjective).  However, we can assume a more intelligent (non-Hippocratic) policy}{1}{-2em} that looks for people with address a1 in the first database (and finds John and Mary), then checks somehow (\eg, by other attributes in the database schema $M_1$ which we did not show) if some of them are recent graduates in IT (and, for example, finds that Mary is such) and thus could work for a company at the address a10 (Google). Then such a policy would propagate update $v_3$ to $v_1$ and $v_2$ as shown in the figure. Now we notice that the composition of updates $u_3$ and $v_3$  is an identity%
\footnote{see the previous footnote~\ref{fnote:oids}},
and thus $u_3;v_3$ would be propagated to the identity on the multimodel \nia. However, $\nia''\noteq\nia$ and Putput fails. (Note that it would also fail if we used a simpler policy of propagating update $v_3$ to identity updates on $A'_1$ and $A'_2$.)

A common feature of the two examples is that the second update $v$ fully reverts the effect of the first one, $u$, while their propagations do not fully enjoy such a property due to some ``side effects''. In fact, Putput is violated if update $v$ would just partially revert update $u$. Several other examples and a more general discussion of this phenomenon that forces Putput to fail can be found in \cite{me-bx17}. However, it makes sense to require compositional update propagation for two {\em sequentially compatible} updates in the sense of \defref{def:mspace}; we call the respective version of the law \kputputlaw\ with 'K' recalling the sequential compatibility relation \Kupd. To manage update amendments, we will also need to use update merging operator \merge\ as shown in \figref{fig:putput}

\begin{defin}[Very well-behaved lenses] \label{def:putput}
	A wb lens is called {\em very well behaved (vwb)}, if it satisfies the following \kputputlaw\ for any $i$, corr $R\in \corr{A_i}$, and updates \frar{u_i}{A_i}{A'_i}, \frar{v_i}{A'_i}{B_i} (see \figref{fig:putput}).

\newpii{	
	Let $\ameui=\ppgr_{ii}(u_i)$. Suppose that  $(u_i,v_i)\in\Kupd_{A_i'}$  and  $(v_i,\ameui)\in\Kdisj{}_{A_i'}$, 
	then the merge $(\tilde v_i, \widetilde{\amex{u_i}}) = \merge(v_i, \ameui)$ is defined (see \figref{fig:putput}) and the following equalities are required to hold:
\\	
	\smallskip
	\noindent \begin{tabular}{l@{\quad}ll}
		\multirow{2}{*}{\lawnamebr{\kputput}{j\noteq i}}  & 
			$\ppgr_{ij}(u_i;v_i)=
		\ppgr_{ij}(u_i);\ppg^{R''}_{ij}(\widetilde{v_i})$, 
		\\ [0.5ex]	& 
		and $(\ppgr_{ij}(u_i),\ppg^{R''}_{ij}(\widetilde{v_i}))\in \Kupd_{A_j''}$,
				\\ [0.5ex] 
			\lawnamebr{\kputput}{ii}	&
		$\amex{(u_i;v_i)} = \widetilde{\amex{u_i}};\amex{(\widetilde{v_i})}$
			\end{tabular}
		}{1: notation is updated by using \amex{\_} symbols}{-9em}
\\[2ex]
see the dashed arrows in \figref{fig:putput}, which depict the propagation of the composed updates. \qed 	
\end{defin}
\begin{figure}
\[
\xymatrix{ 
& \spA_i &&&& \spA_j && \\
&A_i\ar[d]_{u_i}="u_i"\ar@{..}^R[rrrr]& && & A_j\ar[dd]_{\ppgr_{ij}(u_i)}="ppg1" \ar@{-->}@/^3pc/[dddd]^{\ppgr_{ij}(u_i;v_i)} &&\\
&A_i'\ar[dl]_{v_i}="v_i"\ar[dr]^{\ameui
}="ui'"
\ar@{}|{[\merge]}[dd]& && &  &&\\
B_i\ar[dr]|{~~~\widetilde{\ameui}:=\merge_1(v_i,\ameui)}\ar@{-->}@/_2.5pc/[ddr]|{\amex{(u_i;v_i)}%
} &&A_i''\ar[dl]^(0.4){\widetilde{v_i}%
	:=\merge_2(v_i,\ameui)}\ar@{..}^{R''}[rrr] && &  A_j''\ar[dd]_{\ppg^{R''}_{ij}(\widetilde{v_i})}="ppg2"\ar@{}|{=}[r] &&\\
& B_i+_{A_i'}A_i''\ar@{}|(0.7){=}[l]\ar[d]%
^{\amex{\widetilde{v_i}} 
} 
& && &  \\
& A_i'''\ar@{..}^{R'''}[rrrr]& && & A_j''' &&\\
\ar@{-}@/_1.2pc/^{ \in \Kupd} "u_i";"v_i"
\ar@{-}@/_0.5pc/_{ \in \Kdisj{}} "v_i";"ui'"
\ar@{-}@/_1pc/^(0.3){ \in \Kupd_{A_j''}} "ppg1";"ppg2"
}
\]
\caption{\lawnamebr{\kputput}{}-Law\label{fig:putput} (for update $x$, expression $\amex{x}$ stands for the amendment $\ppg^{R(x)}_{ii}(x)$ with $R(x)$ being the corr at the source of $x$ -- it labels the corresponding horizontal dotted line)} 
\end{figure}
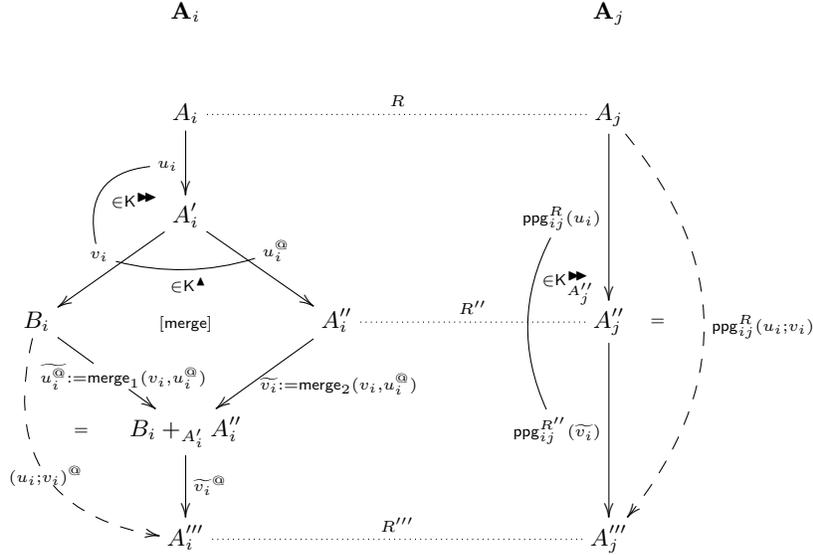

\begin{corollary}[Closed vwb lenses] For a vwb lens as defined above, 
	if update $u_i$ is $R$-closed (\ie, $\ameui=\id_{A'}$), then the following equations hold:

\newpii{ 
		\smallskip
	\noindent \begin{tabular}{l@{\quad}ll}
		\multirow{2}{*}{\lawnameBr{\kputput}{j\noteq i}{closed}}  & 
		$\ppgr_{ij}(u_i;v_i)=
		\ppgr_{ij}(u_i);\ppg^{R''}_{ij}({v_i})$, 
		\\ [0.5ex]	& 
		and $(\ppgr_{ij}(u_i),\ppg^{R''}_{ij}({v_i}))\in \Kupd_{A_j''}$,
		\\ [0.5ex] 
		\lawnameBr{\kputput}{ii}{closed}	&
		$\amex{(u_i;v_i)} = \amex{v_i}$ 
	\end{tabular}
}{2}{-5em}
\end{corollary}

\section{Compositionality of Model \Syncon: Playing Lego with Lenses}\label{sec:lego}
We study how lenses can be composed. Assembling well-behaved and well-tested small components into a bigger one is \zdmodok{the quintessence}{a cornerstone}{} of software engineering.
If a mathematical theorem guarantees that desired properties carry over from the components to the composition, additional integration tests for checking these properties are no longer necessary. This makes lens composition results practically important. 
\par
In Sect.~\ref{sec:lego-para}, we consider parallel composition of lenses, which is easily manageable. Sequential composition, in which different lenses share some of their feet, and updates propagated by one lens are taken and propagated further by one or several other lenses, is much more challenging and considered in the other two subsections. In \sectref{sec:lego-star}, we consider ''star-composition'', cf.\ Fig.~\ref{fig:starComposition} and show that under certain additional assumptions (very) well-behavedness carries over from the components to the composition. However, invertibility does not carry over to the composition -- we shows this with a counterexample in \sectref{sec:invert}. In \sectref{sec:lego-spans2lens}, we study how (symmetric) lenses can be assembled from asymmetric ones and prove two easy theorems on the property preservation for such composition. 
\par
\newpii{
Since we now work with several lenses, we need a notation for lens' components. Given a lens $\ell = (\mspaA, \ppg)$, we write $\spA^\ell, \spA(\ell)$ or $\prtbf\ell$ for $\prtbf\mspA$, $\Corr^\ell$ or $\Corrx{\ell}$ for \Corrx{\mspA}, and $\prt^\ell_i(R)$ for the $i$-th boundary of corr $R$.  Propagation operations of the lens $\ell$ are denoted by $\ell.\ppg_{ij}^{R}$, $\ell.\ppg_{i\star}^{R}$. We will often identify an 
aligned multimodel $(A, R)$ and its corr $R$ as they are mutually derivable (see Remark~\ref{rem:corr-upd} on p.\pageref{rem:corr-upd}). 
}{3}{-2.5cm}

We will also need the notion of lens isomorphism.
\begin{defin}[Isomorphic lenses]\newpii{
	Two $n$-ary lenses $\ell$ and $\ell'$ are {\em isomorphic}, $\ell\cong\ell'$, if 
\\	
	(a) their feet are isomorphic via a family of iso\mor\ functors \frar{f_i}{\spA^\ell_i}{\spA^{\ell'}_i}; 
\\	
	(b) their classes of corrs are isomorphic via bijection \frar{f_\Corr}{\Corr^\ell}{\Corr^{\ell'}} commuting with boundaries: $\prt^{\ell'}_i(f_\Corr(R))=f_i(\prt^\ell_i(R))$ for all $i$; 
\\	
	(c) their propagation operations are compatible with isomorphisms above: for any foot update \frar{u_i}{A_i^\ell}{B_i^\ell} and corr $R\in\Corr^\ell(A_i)$ for lens $\ell$, we have $$f_j(\ell.\ppgr_{ij}(u_i))=\ell'.\ppg^{f_\Corr(R)}_{ij}(f_i(u_i)).$$ That is, two composed mappings, one propagates $u_i$ with lens $\ell$ then maps the result to $\ell'$-space, and the other maps $u_i$ to $\ell'$-space and then propagates it with lens $\ell'$, produce the same result. }{4}{-14em}
\end{defin}

\medskip
\subsection{Parallel Lego: Lenses working in parallel
}\label{sec:lego-para}

\newpii{
	We will consider two types of parallel composition. 
	The first is chaotic (co-discrete in the categorical parlance).  Suppose we have several clusters of \synced\ models, \ie, models within the same cluster are \synced\ but models in different clusters are independent. We can model such situations by considering several lenses $\ell_1,...,\ell_k$, each one working over its own multimodel space $\mspA_i$. Although mutually independent w.r.t data, clusters are time-related and it makes sense to talk about multimodels $\nia_1^t,...,\nia_k^t$ coexisting at some time moment $t$: such a tuple of \mmodel s can be seen as  the state of the multi-multimodel $\nia_1\times\ldots\times\nia_k$ at moment $t$. As model clusters are data-independent, we can propagate tuples of updates $(u_1,\ldots,u_k)$ --- one update per cluster, to other such tuples. For example, if we have a ternary lens $\ekk$ with feet $\spA_i$, $i=1..3$, and a binary lens $\ell$ with feet $\spB_j$, $j=1,2$, any pair of updates 
	$(u_i, v_j)\in\spA_i\timm\spB_j$  can be propagated to pairs $(u_{i'},v_{j'})\in\spA_{i'}\timm\spB_{j'}$ with $i'=1..3$ and $j'=1,2$ by the two lenses working in parallel: lens $\ekk$ propagates $u_i$ and $\ell$ propagates $v_j$. The resulting \syncon\ can be seen as a six-ary lens with feet $\spA_i\times\spB_j$. 
	A general construction that composes lenses $\ell_1,\ldots,\ell_k$ of arities $n_1,\ldots,n_k$ resp., into a product lens $\ell_1\times\ldots\times\ell_k$ of arity $n_1\times n_2\times...\times n_k$ 
	is described in \sectref{sec:lego-para}.1
}{1}{-12em}

\newpii{	
	Our second construct of parallel composition is for lenses of the same arity working in a strongly coordinated way. Suppose that our traffic agency has several branches in different cities, all structured in a similar way, \ie, over metamodels $M_i$, $i=1,2,3$ in \figri. However, now we have families of models $A^x_i$ with $x$ ranging over cities. Suppose also a strong discipline of coordinated updates, in which all models of the same type, \ie, with a fixed metamodel index $i$ but different city index $x$, are updated simultaneously. Then global updates are tuples like $\comprfam{\frar{u^x_1}{A^x_1}{A'^x_1}}%
	               {x\in \mathrm{Cities}}$ or $\comprfam{\frar{u^x_2}{A^x_2}{A'^x_2}}%
	               {x\in \mathrm{Cities}}$. Such tuples can be propagated componentwise, \ie, city-wise, so that we have a global ternary lens, whose each foot is indexed by cities. Thus, in contrast to the chaotic parallel composition, the arity of the coordinated composed lens equals to the arity of components.              
	               We will formally define the construct in \sectref{sec:lego-para}.2. 
}{1 cont'd }{-15em}

\newlength\mygap
\mygap=2ex
\newcommand\gapp{\\[5pt]}
\newcommand\iekk{\ensuremath{{i_\ekk}}}
\newcommand\iell{\ensuremath{{i_\ell}}}
\newcommand\jekk{\ensuremath{{j_\ekk}}}
\newcommand\jell{\ensuremath{{j_\ell}}}

\renewcommand\thissec{\ref{sec:lego-para}}
\renewcommand\mysubsubsection[1]{\par\vspace{-2ex}%
	\subsubsection{\thissec.#1}}
\papertr{
\renewcommand\npmarginn[2]{#1} 
}
{
\renewcommand\npmarginn[2]{#1}
}
\mysubsubsection{1. Chaotic Parallel Composition.}
\begin{defin}\label{def:para-comp-chaotic}
	Let $\ekk$ and $\ell$ be two lenses of arities $m$ and $n$. We first choose the following two-dimensional enumeration of their product $mn$: any number $1\le i \le mn$ 
	is assigned with two natural numbers as specified below:
	\begin{equation}
	i\mapsto (\iekk, \iell) =
	\begin{cases}
	(1,i) \text{ if }  1\le i \le n, 
	\\
	(2,i)  \text{ if }  n < i \le 2n,
	\\
	\ldots
	\\
	(m,i)  \text{ if }  (m{-}1)n < i \le mn,
	\end{cases}
	\end{equation}
Of course, we could choose another such enumeration but its only effect is reindexing/renaming the feet while \syncon\ as such is not affected. 
Now we define the {\em chaotic 
		parallel composition} of $\ekk$ and $\ell$ as the $m\timm n$-ary lens $\kl$ with 
{\renewcommand\item{}
	\begin{tabular}{ll}
		{Boundary spaces:} &
	$\prtbf_i^{\kl} = (\prtbf^\ekk_\iekk, \prtbf^\ell_\iell)$ 
	\gapp
		Corrs: &$\Corrx{\kl} = \Corrx{\ekk}\times\Corrx{\ell}$ with boundaries 
		\\ &		
		$\prt_i^{(\kl)}(Q,R)= (\prt_\iekk^\ekk(Q), \prt_\iell^\ell(R))$ for all $i$
		\gapp
 Operations: & Given an update $\niuu_i$ at foot $i$ of lens $\kl$, \ie, a pair of updates
 \\ 
 & $(u_\iekk,v_\iell)$ with \frar{u_\iekk}{A_\iekk}{A'_\iekk}, \frar{v_\iell}{B_\iell}{B'_\iell}, 
 \\ & and corrs  $Q\in\Corr^\ekk(A_{\iekk})$, $R\in\Corr^\ell(B_{\iell})$, we define  
 \\ & $(\kl).\ppg_{ij}^{(Q,R)}(\niuu_i) \eqdef 
 (\ekk.\ppg_{\iekk\jekk}^Q(u_\iekk), \,\ell.\ppg_{\iell\jell}^R(v_\iell))$
 \\ & 
		$(\kl).\ppg_{i\star}^{(Q,R)}(\niuu_i) \eqdef
 (\ekk.\ppg_{\iekk\star}^Q(u_\iekk), \,\ell.\ppg_{\iell\star}^Q(v_\iell)).$
	\end{tabular}}
	Furthermore, for any models $A\in\xob{\spA}$ and $B\in\xob{\spB}$, relations $\Kupd_{\!A\times B}$ and $\Kdisj{}_{A\times B}$ are the obvious rearrangement of elements of $\Kupd_{A}\times \Kupd_{B}$ and $\Kdisj{}_{A}\times \Kdisj{}_{B}$. 
	\qed
\end{defin}
\begin{lemma} If $\ekk$ and $\ell$ are (very) wb (and invertible), then $\kl$ is (very) wb (and invertible).
\end{lemma}
\proof All verifications can be carried out componentwise.\qed
\begin{lemma}
	Chaotic parallel composition is associative up to isomorphism: $(\ekk\times\ell)\times\ell'\cong\ekk\times(\ell\times\ell')$
\end{lemma}
\proof Straightforward based on associativity of the Cartesian product.\qed

The two lemmas imply
\begin{theorem}[Chaotic Parallel Composition] Let $\ell_1,...,\ell_k$ be a tuple of  lenses of arities $n_i$, $i=1..k$. Then $n_1\timm...\timm n_k$-ary lens $\ell_1\timm...\timm\ell_k$ is defined [up to isomorphism], and it is
	(very) wb (and invertible) as soon as all lenses $\ell_i$ are such.  \qed
\end{theorem}

\newcommand\parsym{\ensuremath{|\!\,|}}
\renewcommand\kl{\ensuremath{\ekk\parsym\ell}}
\mysubsubsection{2. Coordinated Parallel Composition.}
\begin{defin}\label{def:para-comp}
	Let $\ekk$ and $\ell$ be two lenses of the same arity $n$. Their {\em coordinated parallel composition} is the $n$-ary lens $\kl$ with 
	\begin{itemize}
		\item Boundary spaces: $\prtbf_i^{\kl} = (\prtbf_i^\ekk, \prtbf_i^\ell)$ for all $1\le i \le n$, 
		\item Corrs: $\Corrx{\kl} = \xstar{\ekk}\times\xstar{\ell}$ with boundaries $\partial_i^{(\kl)}= (\partial_i^\ekk, \partial_i^\ell)$ for all $i$
		\item Operations: 
		If $\niuu_i=(\frar{u_i}{A_i}{A'_i}, \frar{v_i}{B_i}{B'_i})$ is an update at the $i$-th foot $(A_i, B_i)$ of lens $\kl$, 
		and $Q\in \Corr^{\ekk}(A_i)$, $R\in\Corr^{\ell}(B_i)$ are corrs, then \[(\kl).\ppg_{ij}^{(Q,R)}(\niuu_i) \eqdef (\ekk.\ppg_{ij}^Q(u_i), \ell.\ppg_{ij}^R(v_i))\]
		and 
		\[(\kl).\ppg_{i\star}^{(Q,R)}(\niuu_i) \eqdef (\ekk.\ppg_{i\star}^Q(u_i), \ell.\ppg_{i\star}^R(v_i)).\]
	\end{itemize}
	Furthermore, for any models $A\in\xob{\spA}$ and $B\in\xob{\spB}$, relations $\Kupd_{\!A\times B}$ and $\Kdisj{}_{A\times B}$ are the obvious rearrangement of elements of $\Kupd_{A}\times \Kupd_{B}$ and $\Kdisj{}_{A}\times \Kdisj{}_{B}$.  
	\qed
\end{defin}
\begin{lemma} If $\ekk$ and $\ell$ are (very) wb (and invertible), then $\kl$ is (very) wb (and invertible).
\end{lemma}
\begin{proof}
	All verifications can be carried out componentwise. \qed
\end{proof}
\begin{lemma}
	Coordinated parallel composition is associative up to isomorphism: $(\ekk\parsym\ell)\parsym\ell'\cong\ekk\parsym(\ell\parsym\ell')$
\end{lemma}
\proof Straightforward based on associativity of the Cartesian product.\qed

The two lemmas imply
\begin{theorem}[Coordinated Parallel Composition] Let $\ell_1,...,\ell_k$ be a tuple of  lenses of the same arity $n$.  Then $n$-ary lens $\ell_1\parsym...\parsym\ell_k$ is defined [up to isomorphism], and it is
	(very) wb (and invertible) as soon as all lenses $\ell_i$ are such.  \qed
\end{theorem}
%
%


\subsection{Sequential Lego 1: Star Composition}\label{sec:lego-star}

\newcommand\stepp[1]{
	{\bf Step #1:}}
\renewcommand\stepp[1]{
	{\newline\indent  \bf Step #1:}}

\cellW=7.ex%
\cellH=7.ex
\begin{figure}[h]
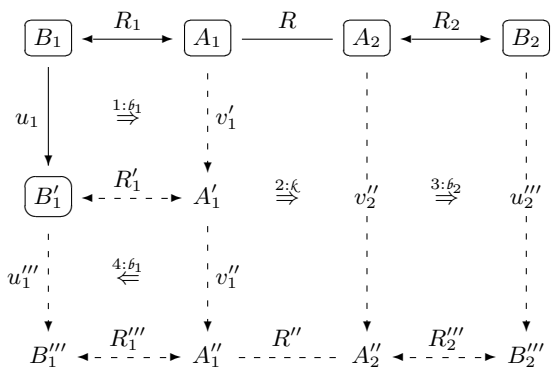

	\centering
	\begin{diagram}[w=1.\cellW,h=1.\cellH]
		\dbox{B_1} & \rCorrto^{R_1} &\dbox{A_1}&\rLine^R&\dbox{A_2}&\rCorrto^{R_2}&\dbox{B_2}
		\\ 
		\dTo<{u_1}&\rtilab{1{:}\ebb_1}&\dDerupdto>{v'_1}&& &&
		\\ 
		\dbox{B'_1}&\rDercorrto^{R'_1}&A'_1&\rtilab{2{:}\ekk}
		             &\dDerupdto~{v''_2}&\rtilab{3{:}\ebb_2}&\dDerupdto~{u'''_2}
		\\ 
		\dDerupdto<{u'''_1}&\ltilab{4{:}\ebb_1}&\dDerupdto>{v''_1}&&&&
		\\ 
		B'''_1&\rDercorrto^{R'''_1}&A''_1&\rDashline^{R''} &A''_2&\rDercorrto{R'''_2}& B'''_2
	\end{diagram}
\caption{Running example via lenses
\label{fig:star-comp-example}}%
\end{figure}
\subsubsection{Running Example Continued.} Diagram in \figriv\ presents a refinement of our example, which explicitly includes relational storage models $B_ {1,2}$ for the two data sources. We assume that object models $A_{1,2}$ are simple projective views of databases $B_{1,2}$: data in $A_i$ are copied from $B_i$ without any transformation, while additional tables and attributes that $B_i$-data may have are excluded from the view  $A_i$.   Synchronisation of bases $B_i$ and their views $A_i$ can be realized by simple {\em constant-complement} lenses $\ebb_i$, $i=1,2$ (see, \eg, \cite{jrw-mscs12}), such that consistent corrs $R_i\in \xstar{\ebb_i}(B_i,A_i)$ (in fact traceability mappings) are exactly those for which the projection of $B_i$ yields $A_i$.

 Finally, let $\ekk$ be a lens \syncing\ models $A_1,A_2,A_3$ as described in \sectref{sec:example}, and $R\in\xstar{\ekk}(A_1,A_2,A_3)$ be a corr for some $A_3$ not shown in the figure.  

Consider the following update propagation scenario. 
Suppose that at some moment we have consistency $(R_1,R,R_2)$ of all five models, \zdnewwok{models $A_{1,2,3}$ are as shown in \figref{fig:multimod-uml} except that model $A_1$ (and the base $B_1$) do not have any data about Mary.}{!!}{-2em} Then model $B_1$ is updated with \frar{u_1}{B_1}{B'_1} that, say, adds to $B_1$ a record of Mary working for Google. 
Consistency is restored with a four-step propagation procedure shown by double-arrows labeled by $x{:}y$ with $x$ the step number and $y$ the lens doing the propagation. %
\stepp{1} lens $\ebb_1$ propagates update $u_1$ to $v'_1$ that adds (Mary, Google) to view $A_1$ with no amendment to $u_1$ as $v'_1$ is just a projection of $u_1$, thus, $B'_1=B''_1$. Note also the updated traceability mapping $R'_1\in \xstar{\ebb_1}(B'_1,A'_1)$
\stepp{2}  lens $\ekk$ propagates $v'_1$ to $v''_2$ that adds (Google, Mary) to $A_2$, and amends $v'_1$ with $v''_1$ that adds (Mary, IBM) to $A'_1$ to satisfy constraint (C1); a new consistent corr $R''$ is also computed.  
\stepp{3} lens $\ebb_2$ propagates $v''_2$ to $u'''_2$ that adds Mary's employment by Google to $B_2$ with, perhaps, some other specific relational storage changes not visible in $A_2$. We assume no amendment to $v''_2$ as otherwise access to relational storage would amend application data. Thus we have a consistent corr $R'''_2$ as shown.
\stepp{4} lens $\ebb_1$ maps update $v''_1$ (see above in Step 2) backward to  $u'''_1$ that adds (Mary, IBM) to $B'_1$ so that $B'''_1$ includes both (Mary, Google) and (Mary, IBM) and a respective consistent corr $R'''_1$ is provided.  
There is no amendment for $v''_1$ by the same reason as in Step 3. 

Thus, all five models in the bottom line of \figriv\ ($A''_3$ is not shown) are mutually consistent and all show that Mary is employed by IBM and Google. \Syncon\ is restored, and we can consider the entire scenario as propagation of $u_1$ to $u'''_2$ and its amendment with $u'''_1$ so that finally we have a consistent corr $(R'''_1,R'',R'''_2)$ interrelating $B'''_1, A''_3, B'''_2$. Amendment $u'''_1$  is compatible with $u_1$ as nothing is undone and condition $(u_1,u'''_1)\in\Kupd_{B'_1}$ holds;  the other two equations required by \lawname{Reflect2-3} for the pair $(u_1,u'''_1)$ also hold. For our simple projection views, these conditions will hold for other updates too, and we have a well-behaved propagation from $B_1$ to $B_2$ (and trivially to $A_3$). Similarly, we have a wb propagation from $B_2$ to $B_1$ and $A_3$. Propagation from $A_3$ to $B_{1,2}$ is non-reflective and done in two steps: first lens $\ekk$ works, then lenses $\ebb_i$ work as described above (and updates produced by $\ekk$ are $\ebb_i$-closed). Thus, we have built a wb ternary lens \syncing\ spaces $\spaB_1, \spaB_2$ and $\spaA_3$ by joining lenses $\ebb_1$ and $\ebb_2$ to the central lens $\ekk$.

\subsubsection{Discussion.} Reflection is a crucial aspect of lens composition. The diagram below describes the scenario above as a transition system and shows that Steps 3 and 4 can be performed in either order. 
\begin{figure}[h]
\centering
	\vspace{-4ex}
	\begin{diagram}[w=2.ex,h=2ex]
		&&&&&&\bullet&&
		\\ 
		&&&&&\ruImplies^{3}&&\rdImplies^{4}&
		\\ 
		\bullet&\rDoubleto_{1}&\bullet&\rDoubleto_{2}&\bullet&&&&\bullet
		\\ 
		&&&&&\rdImplies_{4}&&\ruImplies_{3}&
	     \\ 
	&&&&&&\bullet&&
	\end{diagram}
%
\vspace{-4ex}
\end{figure}

It is the non-trivial amendment created in Step 2 that causes the necessity of Step 4, otherwise Step 3 would finish consistency restoration (with Step 4 being an idle transition ). 
On the other hand, if update $v''_2$ in \figriv\ would not be closed for lens $\ebb_2$, we would have yet another concurrent step complicating the scenario. Fortunately for our example with simple projective views, Step 4 is simple and provides a non-conflicting amendment, but the case of more complex views beyond the constant-complement class needs care and investigation. Below we specify a simple situation of lens composition with reflection a priori excluded, and leave more complex cases for future work.       

\begin{figure}[h]
\vspace{-2ex}
\centering
    \includegraphics[width=0.385\textwidth]%
    {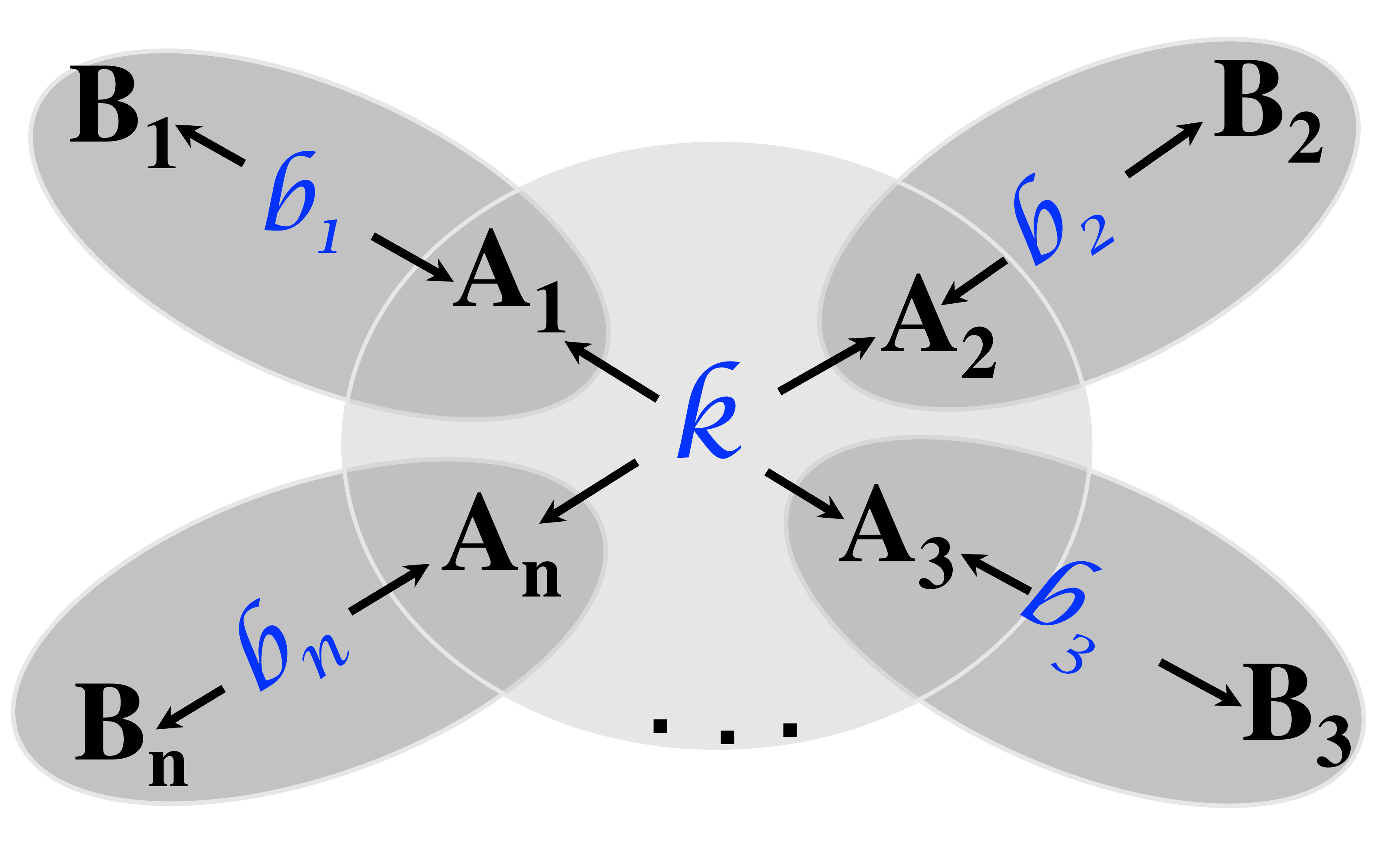} 
\caption{Star Composition}\label{fig:starComposition}
\vspace{-4ex}
\end{figure}
\subsubsection{Formal Definition of Star Composition}\label{def:star-comp}
Suppose we have an $n$-ary lens $\ekk=(\mspaA,\ekk.\ppg)$ with $\mspaA$ based on the model space tuple $\spA = (\spA_1\ldots \spA_n)$,  and for every $i\le n$ a binary lens $\ebb_i =(\spaB_i, \spaA_i, \ebb_i.\ppg)$, with the second model space $\spaA_i$ being the $i$th model space of $\ekk$ (see Fig.\ref{fig:starComposition}, where $\ekk$ is depicted in the center and $\ebb_i$ are shown as ellipses adjoint to $\ekk$'s feet). We also assume that the following {\em Junction Conditions} holds for all $i\le n$:
\begin{equation}\label{eq:junction}
\tag*{\lawnamebr{Junction}{i}}
\mbox{
	\begin{tabular}{l}
   	All updates propagated to $\spaA_i$ by lens $\ebb_i$ are $\ekk$-closed,
	\\ [0.5ex] 
	and all updates propagated to $\spaA_i$ by lens $\ekk$ are $\ebb_i$-closed.
	\end{tabular}
}
\end{equation}

Below we will write the sixtuple of operations $\ebb_i.\ppg^{\acorr_i}$ as the family  \[\comprfam{\ebb_i.\ppg^{\acorr_i}_{xy}}{x\in\{\spA,\spB\}, y\in\{\spA,\spB,\star\}}.\] Likewise we write $\prt^{\ebb_i}_{x}$ with $x\in\{\spaA, \spaB\}$ for the boundary functions of lenses $\ebb_i$.

The above configuration gives rise to the following $n$-ary lens $\ell$. The carrier is the tuple of model spaces $\spaB_1...\spaB_n$ together with their already contained compatible consecutive updates and mergeable updates, resp. Let $B = (B_1\ldots B_n)$ be a model tuple in the carrier, then we define 
\[\xstar\ell(B) = \{(R, R_1...R_n) \mid \exists A = (A_1\ldots A_n) \inn \prod\xob{\spA}: R \in \xstar{\ekk}(A), R_i\in \xstar{\ebb_i}(B_i, A_i)\}\] 
\par 
\zdmoddtok{Note that the disjointment requirement of corr sets in Def.~\ref{def:mspace} yields uniqueness of $A$ in the definition of $\xstar{\ell}(B)$. Thus the realignment function $\xstar{\ell}(u_1\ldots u_n)$ of an update tuple consisting of local updates $\frar{u_i}{B_i}{B_i'}$ can uniqueley be defined: It takes a corr tuple $(R, R_1...R_n)$ and assigns to it the tuple
	\\  
$(R, \xstar{\ebb_1}(u_1,\id_{A_1})...\xstar{\ebb_n}(u_n,\id_{A_n}))$,
\\
i.e.\ it aligns according to idle updates of unique $(A_1\ldots A_n)$ in lens $\ekk$.
}{}{ALIGN}{-2em}
\par

[Consistent] corrs are exactly the tuples $(R, R_1...R_n)$, in which $R$ and all $R_i$ are [consistent] corrs in their respective multispaces. This yields $\prt_i^{\ell}(R, R_1...R_n) = \prt^{\ebb_i}_{\spaB}R_i$ (see Fig.\ref{fig:starComposition}). Propagation operations are defined as compositions of consecutive lens' executions as described below (we will use the dot notation for operation  application and write $x.\mathsf{op}$ for $\mathsf{op}(x)$, where $x$ is an argument). 
 
 Given a model tuple $(B_1...B_n)\in \spB_1\timm...\timm \spB_n$, a corr $(R, R_1...R_n)$, and update \frar{v_i}{B_i}{B'_i} in 
 $\xarr{\spB_i}$,
 we define, first for $j\noteq i$, 
 
 $$
 v_i.\;\ell.\ppg^{(R, R_1...R_n)}_{ij}\;\eqdef\;
 v_i.(\ebb_i.\ppg_{\spB\spA}^{R_i}).
 (\ekk.\ppg^R_{ij}).
 (\ebb_j.\ppg^{R_j}_{\spA\spB}),
 $$
 
 \noindent and then 
 $v_i.\,\ell.\ppg^{(R, R_1...R_n)}_{ii}\;\eqdef\; v_i.\,\ebb_i.\ppg_{\spB\spB}^{R_i}$.
 Note that all internal amendments to $u_i=v_i.(\ebb_i.\ppg_{\spB\spA}^{R_i})$ produced by $\ekk$, and to $u'_j=u_i.(\ekk.\ppg_{ij}^{R})$
 produced by $\ebb_j$, are identities due to the Junction conditions. This allows us to set corrs properly and finish propagation with the three steps above: 
 $v_i.\; \ell.\ppg^{(R,R_1...R_n)}_{i\star}\;\eqdef\;
 (R', R'_1...R'_n)
 $
  where 
 $R'= u_i.\,\ekk.\ppg^R_{i\star}$,  
 $R'_j=u'_j.\,\ebb_j.\ppg^{R_j}_{\spA\star}$ for $j\noteq i$, and $R'_i=v_i.\;\ebb_i.\ppg^{R_i}_{\spB\star}$.
 We thus have a lens $\ell$ denoted by $\ekk^\star(\ebb_1, \ldots, \ebb_n)$. \qed

\begin{theorem}[Star Composition]\label{th:star-comp}
Given a star configuration of lenses as above. Let all underlying model spaces be well-behaved, cf. Def.~\ref{def:wbspace}, and \zdmodtok{. Let lens $\ekk$ fulfill \lawname{Hippocr} and let}{}{} the Junction conditions hold. Then the following holds: 
If lens $\ekk$ and all lenses $\ebb_i$ are (very) wb, then $\ekk^\star(\ebb_1, \ldots, \ebb_n)$ is also (very) wb. 
\zdmodtok{I removed the refined verison of the theorem and replaced it with a simpler one below. The refined version is commneted in the source}{}{}
\end{theorem}
\begin{proof}\zdmoddtok{
For part \ref{th:star-comp1} we need to show that Hippocraticness carries over. Let therefore an update tuple $\overline{u_i} = (\id_{A_1}\ldots u_i\ldots \id_{A_n})$ with $u_i:B_i\to B_i'$ and a consistent corr tuple $(R, R_1\ldots R_n)$ be given. Let the aligned corr $\xstar{\ell}(\overline{u_i})\in \Kcorr$ be consistent. Since it differs from the old one only for the $i$-th component, this means that $\ebb_i.\ppg^{R_i}_{BA}(u_i) = \id_{A_i}$ by hippocraticness of $\ebb_i$. Since $\ekk$ and the other $\ebb_j$ fulfill hippocraticness, they are  stable by corollary \ref{cor:1}, hence further propagations remain identies, as desired.   
}{}{HIPPO}{-2em} 

Fulfilment of \lawname{Stabiliy} is obvious.
\lawname{Reflect1-3} follow immediately from the definition of $\ell.\ppg_{ii}$ above, since the first step of the above propagation procedure already enjoy sequential compatibility and idempotency by \lawname{Reflect1{-}3} for $\ebb_i$. This proves the wb part of the theorem.
\par 
Now we prove that the composed lens is very wb if the components are such. If all model spaces are well-behaved, then operation \merge\ preserves identities, cf. Def.~\ref{def:wbspace}. Then \lawname{KPutPut} for any lens $\ekk$ reduces to the simplified form $\ppg_{ij}^R(u_i;v_i) = \ppg_{ij}(u_i);\ppg_{ij}(v_i)$, if $u_i$ and $v_i$ are $\ekk$-closed. The Junction condition and the fact that operation $\ppg$ preserves the $\Kupd$-property, guarantee that propagation of composed updates from $\spA_i$ to $\spA_j$ and further to $\spB_j$ is not disturbed by reflective updates. Hence, \lawname{KPutPut} carries over from $\ebb_i$ to $\ekk^\star(\ebb_1, \ldots, \ebb_n)$.  
\qed
\end{proof}

\subsection{Lens composition and Invertibility} \label{sec:invert}
\renewcommand\thissec{\ref{sec:invert}}
\renewcommand\mysubsubsection[1]{\par\vspace{-1ex}%
	\subsubsection{\thissec.#1}}

Unfortunately, even if all component lenses are invertible, the composed star-lens is not necessarily such as we will show in the next subsection. In subsection 6.3.2, we discuss a seemingly counter-example to this negative result 
and show that the state-based setting for update propagation can be confusing.   
\mysubsubsection{1 Counter-example.} Consider a class of simple model spaces, whose objects (called models) are natural numbers plus some fixed symbol $\bot$ denoting an undefined number. 
Thus, for a space \spA\ from this class, we assume $\xob{\spA}=\{\bot\}\cup \spA!$ with $\spA!=\{A_1,A_2,{...}\}$  being a set of natural numbers: $A_i=\{i\}$.  Models form the set $\spA!$ are called {\em certain} while models $\bot$ are {\em uncertain}. Updates are all possible pairs of models: $\xarr{\spA}=\xob{\spaA}\timm\xob{\spA}$ (such categories are often called co-discrete or chaotic), that is,
\renewcommand\Cup{~\cup~} 
$$\xarr{\spA}=\{(\bot,\bot)\}
\Cup \{\bot\}\timm \spA!
\Cup \spA!\timm \{\bot\}  
\Cup \spA!\timm \spA!     
$$ 
Now consider three model spaces of the type specified above: space $\spB_1$ with certain models $B_1=\{0,2,6\}$, space $\spA$ with certain models $A=\{2,5,6\}$ and space $\spB_2$ with $B_2=
\{4,7\}$. For all these spaces, mergeable pairs are only those with at least one identitity and the same necessarily holds for $\Kupd$, hence all spaces are well-behaved (cf. \defref{def:wbspace}).

   Suppose we have a lens $\ebb_1$ over spaces $\spB_1$ and $\spA$ (see \figref{fig:counterex}). The set of corrs for a pair $(b,a)\in \xob{\spB_1}\timm \xob{\spA}$ is the singleton set $\{(b,a)\}$. A corr is consistent, if and only if $b,a$ are both certain or both uncertain, \ie, consistency amounts to equal certainty. \zdmodtok{Alignment is obvious (because all corr sets are singletons).}{}{} Then to restore consistency, updates of type $B_1\timm\{\bot\}$ are to be mapped to similar updates of type $A\timm\{\bot\}$, and updates of type $B_1\timm B_1$ can be mapped to the corresponding identity updates on the $A\timm A$ as they do not destroy consistency. The propagation operation from \spA\ to $\spB_1$ of the types above are defined similarly. 
   
    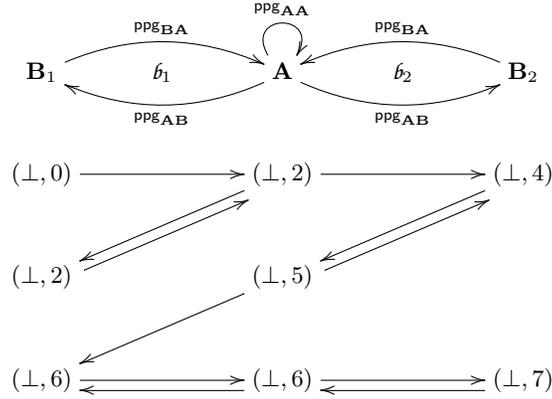
\begin{figure}
\[\xymatrix{
\spB_1 \ar@/^1pc/[rr]^{\ppg_{\spB\spA}}&\ebb_1& 
\spA \ar@/_1pc/[rr]_{\ppg_{\spA\spB}}\ar@(ul,ur)^{\ppg_{\spA\spA}}
	\ar@/^1pc/[ll]^{\ppg_{\spA\spB}}
&\ebb_2& \spB_2\ar@/_1pc/[ll]_{\ppg_{\spB\spA}} \\
(\bot,0) \ar[rr]&& (\bot,2)\ar[dll]\ar[rr] && (\bot,4)\ar[dll] \\
(\bot,2) \ar@<-1.ex>[urr] && (\bot,5)\ar@<-1.ex>[urr]\ar[dll] &&  \\
(\bot,6) \ar[rr]&& \ar@<1.ex>[ll] (\bot,6)\ar[rr] && (\bot,7)\ar@<1.ex>[ll] \\
}
\]
\caption{Schematic description of the counterexample}\label{fig:counterex}
\end{figure}
   Now we need to define how to propagate updates of the type $\{\bot\}\timm B_1$ over the corr $\{(\bot,\bot) \}$. 
 
 This is specified in \figref{fig:counterex} by arrows going from the left-most column of updates to the middle column. The idea is to find in \spA\ the nearest model so that updating $\bot$  to 0 or 2 in $\spB_1$ goes to updating $\bot$ to 2 in \spA, while updating $\bot$ to 6 goes to updating $\bot$ to 6. In the similar way, propagation of updates from $\{ \bot \}\timm A$ to $\{\bot\}\timm B_1$ is defined as shown by arrows from the middle to the left column. This defines lens $\ebb_1$, and it is easy to see it is well-behaved and, moreover, invertible. The latter is established by the direct examination of all possible $\ppg$-compositions, \eg, 
 
 \smallskip
 $(0,\bot).\ppg^{(0,2)}_{\spB\spA}
 .\ppg^{(0,2)}_{\spA\spB}
 =(0,\bot)$ (even strong invertibility holds)
 
 \medskip
  $(0,2).\ppg^{(0,2)}_{\spB\spA}
 .\ppg^{(0,2)}_{\spA\spB}
 =(0,0)\ne (0,2)$ but $(0,2).\ppg^{(0,2)}_{\spB\spA}=(0,0).\ppg^{(0,2)}_{\spB\spA}$
 
 \medskip 
 $(\bot,0).\ppg^{(\bot,\bot)}_{\spB\spA}
 	.\ppg^{(\bot,\bot)}_{\spA\spB}=(\bot,2)\ne (\bot,0)$
 	but
 	$(\bot,0).\ppg^{(\bot,\bot)}_{\spB\spA}  
 	= (\bot,2).\ppg^{(\bot,\bot)}_{\spB\spA}$
 
 \medskip	
\noindent (the last example is specific for the lens $\ebb_1$ specified in \figref{fig:counterex}). Similarly, we define a wb lens $\ebb_2$ as shown in \figref{fig:counterex} and check it is also invertible.

Space $\spA$ in the middle can be extended to a trivially wb and invertible identity lens $\ekk = \idlensna$ with $n=2$ (cf.\ Example \ref{ex:idLens}). 

Now, as the Junction condition trivially holds for the triple $(\ebb_1,\ekk,\ebb_2)$, we obtain a star lens $\ell:=\ekk^\star(\ebb_1, \ebb_2)$ composed from wb invertible lenses. However, this lens is not invertible as the following computation shows:  
\[
(\bot,0).(\ell.\ppg_{12}).(\ell.\ppg_{21}).(\ell.\ppg_{12}) = (\bot,7) \ne (\bot,4) = (\bot,0).(\ell.\ppg_{12})
\]
where we omitted the upper scripts $(\bot,\bot)$ near \ppg-symbols (recall that $\ell.\ppg_{12}$ is defined as the sequential composition 
$(\ebb_1.\ppg_{\spB\spA}).
(\ekk.\ppg_{\spA\spA}).
(\ebb_2.\ppg_{\spA\spB})$ and similarly for $\ell.\ppg_{12}$).
\begin{theorem} 
	Star-composition does not preserve invertibility
\end{theorem}


	\mysubsubsection{2 Invertibility and (binary) state-based symmetric lenses with complement (ssc-lenses): a long standing confusion.}

	The example above may seem to be contradicting to paper \cite{bpierce-popl11}, where symmetric lenses are studied in the state-based setting as ssc-lenses. In that paper, an invertibility law called {\em round-tripping} is required for any ssc-lens, and it is proved that sequential composition of such lenses is again an ssc-lens and hence enjoys round-tripping. In paper \cite{me-models11}, we show that an ssc-lens is nothing but a symmetric delta lens over co-discrete model spaces (see also \cite{jr-bx16}), \ie, exactly a lens of the type we have considered in the counter-example above. Moreover, in the star-composition instance we have considered, lens $\ekk=\idlens{2}{\spA}$ plays a dummy role and, in fact, we have dealt with sequential composition of two ssc-lenses $\ebb_1;\ebb_2$ as defined in \cite{bpierce-popl11}. Then, how could it happen that our composed lens does not satisfy (even weak) invertibility, while the corresponding result in \cite{bpierce-popl11} asserts that the composed lens must satisfy the seemingly stronger roundtripping law?  
	
	The source of confusion is the state-based setting for update propagation, in which a law that looks like demanding strong invertibility is actually a simple \lawname{Stability} law demanding identity preservation.
	 Indeed, in \cite[Def.2.1 on p.2]{bpierce-popl11}, they define a binary symmetric lens over state spaces $X$ and $Y$ with propagation operations 
	\begin{equation}\label{eqLputLR}
	\frar{\ppg_{XY}}{X\times\Corr}{Y\times\Corr} \mbox{ and } \flar{\ppg_{YX}}{X\times \Corr}{Y\times\Corr} 
	\end{equation} (they call elements of set \Corr\ complements rather than corrs, but as it is shown in \cite{me-models11,jr-bx16}, the two notions are equivalent). The law (PutRL)  called {\em round-tripping} is defined thusly: for any states $x$ and $y$ (in our notation,  $A'_1$ and $A'_2$) and complements $c,c'$ (\ie, our corrs, $R,R'$), the following condition holds (\cite[Def.2.1 on p.2]{bpierce-popl11}):
	\begin{equation}\label{eq:false-round}
	\ppg_{XY}(x,c) = (y,c') \mbox{ implies } \ppg_{YX}(y,c')=(x,c')
	\end{equation}

\begin{figure}
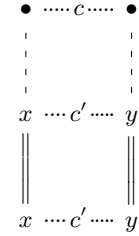

	\begin{minipage}{0.735\textwidth}
			Figure \ref{fig:false-round} specifies the story described by the two equations in \eqref{eq:false-round} diagrammatically: the upper square specifies data of the left equation, and the lower square specifies the right equation. Now it is seen that the right equation specifies the Stability law of delta lenses, while the left equation is used to ensure that states $x$ and $y$ are related by complement/corr $c'$. 
		The left equation is needed as the state-based lens \fwk\ does not have primitives to say that states $(x,y)$ are consistent via corr $c'$ (which is done in delta lenses by boundary functions $\prt_i$), and thus they encode this fact equationally, which together with the right equation creates a false impression of a round-tripping law. Indeed, equations \eqref{eq:false-round} have nothing to do with roundtripping: to describe the latter, we should specify how the result of $\ppg_{YX}(y,c)$ is related to $x$ rather than how $\ppg_{YX}(y,c')$ is related to $x$. 
	\end{minipage}
\hspace{1ex}
\begin{minipage}{0.215\textwidth}
		\vspace{-6ex}
		\centering
	\begin{diagram}[small]
		\bullet &\rDots~{c} &\bullet
		\\
		\dDashline &&\dDashline
		\\
		x &\rDots~{c'}& y
		\\
		\dDLine &&\dDLine
		\\
		x &\rDots~{c'}& y
  	\end{diagram} 
		\caption{``Round- tripping'' story}\label{fig:false-round}
		\vspace{-1ex}
\end{minipage}
\vspace{-4ex}
\end{figure}
%
The main problem is not in that Stability is called Round-tripping, the problem is that the issue of an actual round-tripping law and its  preservation under lens composition is not stated in the state-based lens \fwk\ because it is (fallaciously) considered solved! The counter-example presented above makes the problem even more challenging as even weak round-tripping is not preserved under lens composition. 	
	%

\subsection{Sequential Lego 2: Assembling $n$-ary Lenses from Binary Lenses}\label{sec:lego-spans2lens}
This section shows how to assemble $n$-ary (symmetric) lenses from binary asymmetric lenses modelling view computation \cite{me-jot11}. As the latter is a typical bx, the well-behavedness of asymmetric lenses has important distinctions from well-behavedness of general (symmetric mx-tailored) lenses.
\begin{defin}[Asymmetric Lens, cf. \cite{me-jot11}]
	An {\em asymmetric lens (a-lens)} is a tuple $\ebb^{\preccurlyeq}=(\spaA,\spaB,\get, \putl)$ with \spaA\ a model space called the {\em (abstract) view}, \spaB\ a model space called the {\em base}, \flar{\get}{\spaA}{\spaB} a functor (read ``get the view''), and \putl\ a family of operations $\comprfam{\putl^B}{B\in\xob{\spaB}}$ (read ``put the view update back'') of the following arity. Provided with a view update \frar{v}{\get(B)}{A'} at the input, operation $\putl^B$ outputs a base update \frar{\putl^B_\indb(v)=u'}{B}{B''}  
	and an amendment to the view update, $\putl^B_\indb(v)$ or \frar{\amex{v}}{A'}{A''}. 

A view update \frar{v}{\get(B)}{A'} is called {\em closed} if 
$\amex{v}=\id_{A'}$, and an a-lens is {\em closed} if all its view updates are closed.  \qed
\end{defin}
The following is a specialization of Def.~\ref{def:wblens}, in which consistency between a base $B$ and a view $A$ is understood as equality $\get(B)=A$.

{ 
	\renewcommand\ppgB{\ensuremath{\putl^B}}
	\begin{defin}[Well-behavedness]\label{def:wbalens}
		An a-lens is {\em well-behaved (wb)} if it satisfies the following laws for all $B\in\xob{\spaB}$ and \frar{v}{\get(B)}{A'}
			\\[1ex]
		\lawgap=1ex
		\noindent \begin{tabular}{l@{\quad}l}
			\lawnamebr{Stability}{
			}	& 
			if $v=\id_{\get(B)}$, then $\ppgB_\indb(v) = \id_B$ and  
			$\amex{v}=\id_{\get(B)}$
							\\  [\lawgap] \lawnamebr{\reflectzero}{
			} &	
			if $v=\get(u)$ for some \frar{u}{B}{B'}, then $\amex{v}=\id_{A'}$
			\\  [\lawgap] \lawnamebr{Reflect1}{
			} &	
			$(v, \amex{v})\in \Kupd_{A'}$ 
			\\  [\lawgap] \lawnamebr{Reflect2}{
			} &
			$\ppgB_{\indb}(v;\amex{v}) = \ppgB_{\indb}(v)$
			\\ [\lawgap] \lawnamebr{PutGet}{
				} & 
			$v.\putl_\indb^B.\get =v;\amex{v} $ %
			\footnotesize{(the dot notation is used to highlight the name of the law).}   
			\end{tabular} \\
		\rule{0mm}{0mm}\qed
	\end{defin}
}

%

In contrast to the general lens case, a wb a-lens features \lawname{\reflectzero} --- a sort of a self-Hippocratic law regulating the necessity of amendments: the latter is not allowed as soon as consistency can be restored without amendment. Another distinction is inclusion of an invertibility law \lawname{PutGet} into the definition of well-behavedness: 
\lawname{PutGet}{} together with \lawname{Reflect2} provide (weak) invertibility: $\putl^B_\indb(\get(\putl^B_\indb(v)))=\putl^B_\indb(v)$.   
\lawname{Reflect3} is omitted as it is implied by \lawname{Reflect0} and \lawname{PutGet}.  

Any a-lens $\ebb^{\preccurlyeq} =(\spaA,\spaB,\get, \putl)$ gives rise to a binary symmetric lens $\ebb$. Its carrier consists of model spaces $\spaA$ and $\spB$, for which we assume $\Kupd$ and $\Kdisj{}$ to be defined such that the spaces are well-behaved, cf.\ Def.~\ref{def:wbspace}. The set of corrs of a pair $(A,B)\in \xob{\spA}\times \xob{\spB}$ is the singleton $\{\etosing{(A,B)}\}$ with $\prt_\spA\etosing{(A,B)}=A$ and $\prt_\spB\etosing{(A,B)}=B$, 
\zdmodtok{realignment is thus trivial,}{}{ALIGH} and consistent corrs are exactly those \etosing{(A,B)}, for which $A = \get(B)$.

\renewcommand\etosing[1]{\elementToSingleton{#1}}
For a consistent corr $\etosing{B}={(\get(B), B)}$, 
we need to define six operations $\ebb.\ppghB_{\_\_}$. Below we will write the upper index as $B$ rather than $\hat B$ to ease the notation. For a  view update \frar{v}{A}{A'} and the corr \etosing{B}, we define 
 $$
\ppgB_{\spA\spB}(v)=\putl^B_\indb(v){:}\, B\to B'', \;
\ppgB_{\spA\spA}(v)=\putl^B_\inda(v){:}\, A'\to A'', \;
\ppgB_{\spA\star}(v)=\etosing{B''}
$$
 The condition $A''=\get(B'')$ for $\ebb^\preccurlyeq$ means that \etosing{B''} is again a consistent corr with the desired boundaries. For a base update \frar{u}{B}{B'} and the corr \etosing{B}, we define
$$\ppgB_{\spB\spA}(u)=\get(u),\;
\ppgB_{\spB\spB}(u)=\id_{B'},\; 
\ppgB_{\spB\star}(u)=\etosing{B'}
$$ Functoriality of $\get$ yields consistency of $\widehat{B'}$.

\begin{defin}[Very well-behaved a-lenses] \label{def:putput_a}
		A wb a-lens is called {\em very well behaved (very wb)}, if the \corring\ binary symmetric lens is such.  
		cf.  Def.~\ref{def:putput}.
Of course, this definition restricts the behaviour only for the $\putl$-part of the a-lens (hence the name of the law), since functor $\get$ is compositional anyway. \qed
\end{defin}

\begin{lemma}\label{lemma:asymm_symm} Let $\ebb^\preccurlyeq$ be a (very) wb a-lens and $\ebb$ the corresponding symmetric lens. Then all base updates of $\ebb$ are closed, and $\ebb$ is (very) wb and invertible. 
 \end{lemma}
 \begin{proof}
 Base updates are closed by the definition of $\ppg_{\spB\spB}$. Well-behavedness follows from wb-ness of $\ebb^\preccurlyeq$. Invertibility has to be proved in two directions: $\ppg_{\spaB\spaA};\ppg_{\spaA\spaB};\ppg_{\spaB\spaA} = \ppg_{\spaB\spaA}$ follows from \lawnamebr{PutGet}{} and \lawnamebr{Reflect0}, the other direction follows from \lawnamebr{PutGet}{} and \lawnamebr{Reflect2}, see the remark after Def.\ref{def:wbalens}. The ''vwb''-implication follows directly from Def.\ref{def:putput_a}. 
 \qed
 \end{proof}
\begin{remark}[A-lenses as view transformation engines]\newpii{
	In terms of Varr\'o \etal\ \cite{dvarro-models18}, a wb a-lens can be seen as an abstract algebraic model of a bidirectional view transformation engine. This engine is assumed to be a) consistent: it is a basic lens law that propagation must always produce a consistent corr, b) incremental, as lenses propagate changes rather than create the target models from scratch, c) validating, as the result of propagation is always assumed to satisfy the target metamodel constraints (either by considering the entire model space to be only populated by valid models, or by introducing a subclass of valid models and require the result of the propagation to get into that subclass). Reactiveness, \ie, whether the engine executes on-demand or in response to changes, is beyond the lens \fwk. To address this and similar concerns, we need a richer \fwk\ of lenses augmented with an {\em organizational} structure introduced in \cite{me-jss15} (see also \cite{me-grand17} for a concise presentation).}{1}{-15em}
\end{remark}

  \begin{theorem}[Lenses from Spans]\label{thm:LensesFromSpans} An $n$-ary span of (very) wb asymmetric lenses $\ebb^{\preccurlyeq}_i=(\spaA_i, \spaB,$ $\get_i, \putl_i)$, $i=1..n$ with a common base \spaB\ of all $\ebb_i^{\preccurlyeq}$ gives rise to a (very) wb symmetric lens denoted by $\Sigma_{i=1.. n}^\spB\ebb_i^{\preccurlyeq}$. 
\end{theorem}
\begin{proof}
An $n$-ary span of a-lenses $\ebb_i^{\preccurlyeq}$ (all of them interpreted as symmetric lenses $\ebb_i$ as explained above) is a construct equivalent to the star-composition of Def.~\ref{def:star-comp}.3, in which lens $\ekk=\idlensnb$ (cf. Ex.~\ref{ex:idLens}) and peripheral lenses are lenses $\ebb_i$. The junction condition is satisfied as all base updates are $\ebb_i$-closed for all $i$ by Lemma \ref{lemma:asymm_symm}, and also trivially closed for any identity lens. The theorem thus follows from Theorem~\ref{th:star-comp}. The ``very wb''-part follows directly from Def.~\ref{def:putput_a}. Note that a corr in $\Corrx{\Sigma_{i=1.. n}^\spB\ebb_i^{\preccurlyeq}}$ 
 is nothing but a single model $B\in\xob{\spaB}$ with boundaries being the respective $\get_i$-images.
\qed
\end{proof}
The theorem shows that combining a-lenses in this way yields an $n$-ary symmetric lens, whose properties can automatically be inferred from the binary a-lenses.

\begin{figure}
\centering
    \includegraphics[width=0.5\textwidth]%
                 {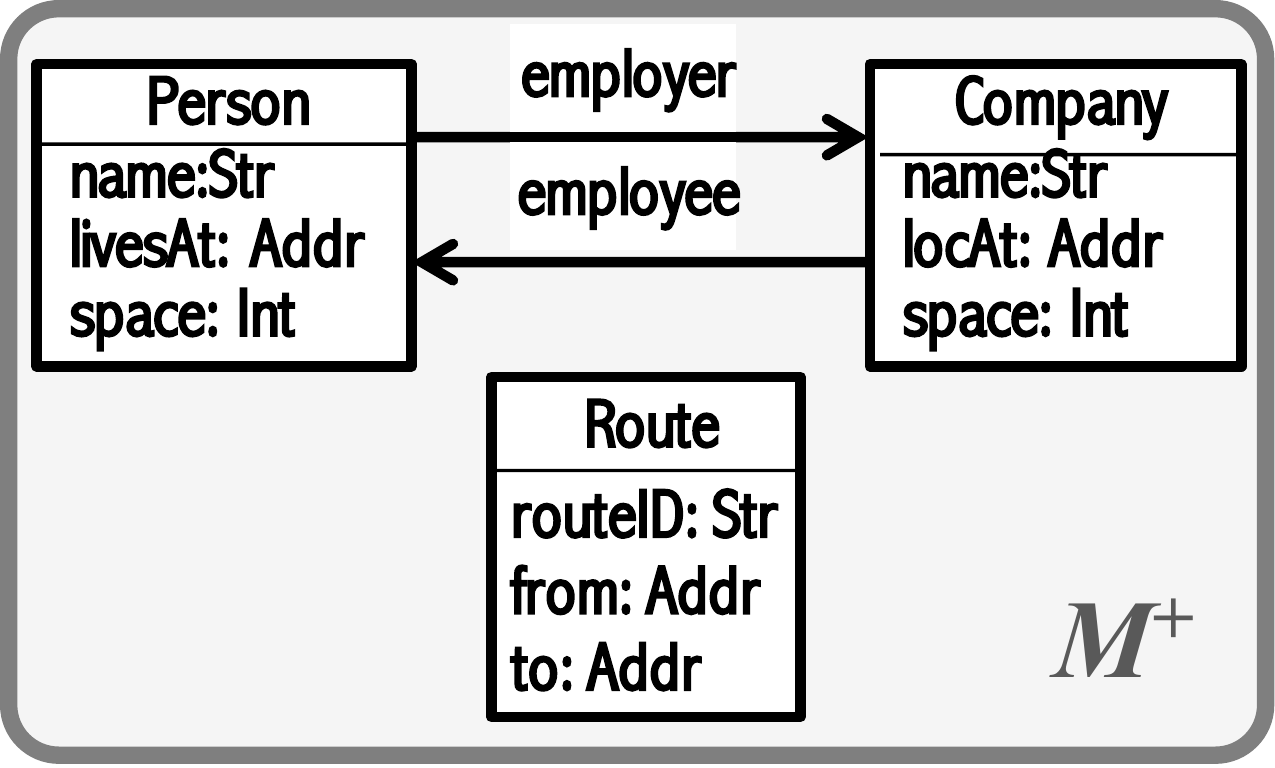}
\caption{Merged Metamodel\label{fig:mergedMetamod-uml}}
\vspace{-4ex}
\end{figure}
\noindent {\em Running example.} Figure~\ref{fig:mergedMetamod-uml} shows a metamodel $M^+$ obtained by merging the three metamodels $M_{1,2,3}$ from \figri\ without loss and duplication of information. In addition, for persons and companies, the identifiers of their  original model spaces can be traced back via attribute ''space'' 
(\Commute-objects are known to appear in space $\spA_3$ and hence do not need such an attribute). As shown in \cite{me-ecmfa17}, any consistent multimodel $(A_1...A_n, R)$ can be merged into a comprehensive model $A^+$ instantiating $M^+$. Let 
$\spB$ be the space of such together with their comprehensive updates \frar{u^+}{A^+}{A'^+}. 

For a given $i\le 3$, we can define the following a-lens $\ebb^{\preccurlyeq}_i=(\spaA_i, \spaB,$ $\get_i, \putl_i)$:
$\get_i$ takes update $u^+$ as above and outputs its restriction to the model containing only objects recorded in space $\spA_i$. Operation $\putl_i$ takes an update \frar{v_i}{A_i}{A'_i} and first propagates it to all directions as discussed in \sectref{sec:example}, then merges these propagated local updates into a comprehensive  \spaB-update between comprehensive models. This yields a span of a-lenses that implements the same \syncon\ behaviour as the symmetric lens discussed in \sectref{sec:example}.

\noindent {\em\bf  From (symmetric) lenses to spans of a-lenses.} There is also a backward transformation of (symmetric) lenses to spans of a-lenses. Let $\ell=(\mspaA, \ppg)$ be a (very) wb lens. It gives rise to the following span of (very) wb a-lenses $\ell_i^\preccurlyeq=(\prtbf_i(\mspaA),\spaB,\get_i, \putl_i)$   
where space $\spaB$ has objects $(A_1\ldots A_n,R)$ with $R\in \Corrx{A_1...A_n}$ a consistent corr, and arrows are the update tuples of $\mspA$. Functors $\get_i:\spaB\to \spaA_i$ are projection functors. Given $B = (\tupleAWOB, R)$ and local update \frar{u_i}{A_i}{A'_i}, let

\frar{\putl^B_{i\indb}(u_i)\eqdef(u_1',..,u'_{i-1}, (u_i;u_i'), u'_{i+1},..,u_n')}
{(A_1... A_n,R)}{(A''_1... A''_n,R'')}
\\[1ex] 
where $u_j' \eqdef \ppg^R_{ij}(u_i)$ (all $j$) and $R'' = \ppg_{i\star}^R(u_i)$. Finally, $\putl^B_{i\inda}(v_i)\eqdef\ppg^R_{ii}(v_i)$, \ie, 
$v_i^{@B}=v_i^{@R}$. 
Validity of \zdmodtok{\lawname{Hippocr1-2}}{\lawname{Stability}}{}, \lawname{Reflect0-2}, \lawname{PutGet} and \kputputlaw\ directly follows from the above definitions. This yields 
\begin{theorem}\label{theo:lens2span}
Let $\ell = (\mspA, \ppg)$ be a (very) wb symmetric lens. Then the multi-span 
\[(\ell_i^\preccurlyeq=(\prtbf_i(\mspaA),\spaB,\get_i, \putl_i) \mid i=1..n)\]
of (very) wb a-lenses (where \spB\ is defined as specified above) has the same synchronisation behaviour w.r.t.\ model spaces $(\prtbf_i(\mspaA)\mid i=1..n)$.  
\qed
\end{theorem}

\newpii{
An open question is whether the span-to-lens transformation in Thm.\ref{thm:LensesFromSpans} and the lens-to-span transformation of Thm.~\ref{theo:lens2span} are mutually inverse. It is easy to see that the chain 
\begin{equation}\label{eq:lens2span2lens}
\mbox{symmetric lens $\rightarrow$ wide span $\rightarrow$ symmetric lens} 
\end{equation}
results in a lens isomorphic to the initial one, but the chain 
\begin{equation}
\mbox{wide span $\rightarrow$ symmetric lens $\rightarrow$ wide span} 
\end{equation}
ends with a different span as the first transformation above loses information about updates in the head of the span. 
As shown by Johnson and Rosebrugh in \cite{jr-jot17} for the binary case, the two spans can only be equivalent modulo certain equivalence relation, and an equivalence relation between lenses is also needed to align all constructs together. These equivalences may be different for our multiary lenses with amendments, and we leave this important question for future research.  
}{1}{-17em}

\section{Future work}\label{sec:future}

\renewcommand\thissec{\ref{sec:future}}
\renewcommand\mysubsubsection[1]{\par\vspace{-1.5ex}%
	\subsubsection{\thissec.#1}}

\newpii{We list and briefly comment on several important tasks for the multiary lens \fwk.}{2:this section is essentially reworked.}{-2em}

\subsection{Categorification of corrs} 
\newcommand\scat{\ensuremath{\mathbf{S}}}

A distinctive feature of the \fwk\ developed in the paper is the triviality of the corr updates --- they are just pairs $(R,R')$ of the old and the new corr (see Remark~\ref{rem:corr-upd} on p.\pageref{rem:corr-upd}). However, in practice, new corrs would be computed incrementally with deltas rather than afresh. To make the \fwk\ closer to practice, we need to change the notion of a multimodel update \frar{\mmu}{\nia}{\nia'} and consider it to be a pair $\mmu=(u,r)$ with $u=(u_1..u_n)$ a feet update and \frar{r}{R}{R'} a corr update rather than just a pair of states $(R,R')$.
Then we would obtain a setting based on a category \mspR\ of \mmodel s (that includes both local updates and corr updates) together with boundary projection functors \frar{\prtbf_i}{\mspR}{\spA_i}, $i=1..n$, which take a multimodel update $\mmu=(u_1...u_n, r)$ and select its corresponding component, $\prtbf_i(\mmu)=u_i$.%
\footnote{Given some construct $X$ depending on set $S$, the passage to a setting in which $S$ is a non-trivial category is often called {\em categorification} of $X$, hence, the title of the subsection. The quality of being a {\em non-trivial} category is essential. A set $S$ can be seen as a category \scat\ in two ways: {\em discrete}, when the only \scat\ arrows are identities, and {\em co-discrete {\em or} chaotic}, in which for any pair of $S$'s elements $s,s'$, there is one and only one arrow \frar{(s,s')}{s}{s'}. If $s'=s$, this arrow is the identity of $s$. We refer to both such categories \scat\ as trivial.}

For multiary update propagation, we require each projection $\prtbf_i$ to be the \get-part of an asymmetric lens (with amendment) that for a given multimodel state $R$,  puts any foot update \frar{u_i}{A_i}{A'_i}, $A_i=\prtbf_i(R)$ back to a multimodel update \frar{\mmu=\mmput^R_i(u_i)}{\nia}{\nia'}. Note that as \mmu\ includes all feet updates, operation $\mmput^R_i$ actually provides all local propagation operations $\ppg_{ij}$ we considered in the paper: $\ppg^R_{ij}(u_i)=\prtbf_j(\mmput^R_i(u_i))$. 
The Putget law is the equality $\prtbf_i(\mmu)=u_i;u^@_i$. In this way, a multiary symmetric lens is---by definition---a multiary span of asymmetric lenses \frar{(\prtbf_i,\mmput_i)}{\mspR}{\spA_i}, $i=1...n$. 

Besides being better aligned with practice, the setting above would  simplify notation and probably some technicalities, but its advantages were recognized when the paper was already written and submitted for reviewing. With a great regret, an accurate theory of the categorified version of multiary lenses is thus left for future work.

\mysubsubsection{2 Instantiation of the \fwk.} 
Although we presented several simple examples of multiary wb lenses, having practically interesting examples, or even better, a pattern for generating practically interesting examples, is an extremely important task. We think that the GDG \fwk\ for \syncon\ \cite{frankT-icmt15} is a promising foundation for such work, and it has actually already began in \cite{frankT-icmt16} (see the next section for more detailed comments). 

\mysubsubsection{3 Concurrent updates.}
As discussed in section \sectref{sec:backgr}.4, extending the lens  formalism to accommodate concurrent updates is both useful and challenging. One of the main problems to solve is how to reconcile non-determinism inherent in conflict resolution strategies and serialization of update propagation with the deterministic nature of the lens \fwk. 
%

\mysubsubsection{4 Richer compositional \fwk.}
There are several open issues here even for the non-concurrent case (not to mention its future concurrent generalization). First, our pool of lens composition constructs is far incomplete: we need to enrich it, at least, with (i) sequential composition of a-lenses with amendment so that a category of a-lenses could be built, and (ii) a relational composition of symmetric lenses sharing several of their feet (similar to relational join). It is also important to investigate composition with weaker junction conditions than that we considered. 

\mysubsubsection{5 Invertibility.}
Invertibility nicely fits in some but not all of our results. It is a sign that we do not well understand the nature of invertibility. Perhaps, while invertibility is essential for bx, its role for mx may be less important. 

\mysubsubsection{6 Hippocraticness.}
A very natural Hippocraticness requirement (introduced by Stevens for bx in \cite{stevens-sosym10}) may have less weight in the mx world. Indeed, the examples we considered in the paper show that non-Hippocratic consistency restoration may be practically reasonable and useful. A further research in this direction is needed.  


\section{Related Work}\label{sec:related} 
Compositionality as a fundamental principle for building \syncon\ tools was proposed by Pierce and his coauthors in \cite{foster07,boomerang-08} for the state-based asymmetric case, and further developed for the binary symmetric case by Hofmann \etal: in \cite{bpierce-popl11}, for the state-based setting, and in \cite{bpierce-popl12}, for a specific delta-based setting (where deltas are understood operationally as edits). The deficiencies of the state-based \fwk\ are discussed in detail in \cite{me-jot11,me-models11} (see also our \sectref{sec:invert}.2), where the notion of asymmetric and symmetric binary delta lenses were proposed, but lens composition was not specially considered (except easy sequential composition of a-lenses in \cite{me-jot11}). In section 6.1 of \cite{me-gttse09}, several results for delta lens composition are considered, including a special case of the star composition, when two asymmetric lenses form a cospan, but the consistency relation at the apex of this cospan is not assumed to be an equality and, hence, to be maintained by a general binary symmetric lens. 
%
A fundamental theory of composing (and decomposing) binary symmetric lenses from (into) spans or cospans of asymmetric ones was developed by Johnson and Rosebrugh \cite{jr-jot17,jr-bx18}; the theory is based on equivalences of lenses w.r.t. their behaviour. In all these works, 
neither  amendments nor intelligent (and hence constrained) versions of Putput are considered (and only binary lenses are taken into consideration).

\newcommand\tral{Trollmann and Albayrak}
\newcommand\tralmm{\ensuremath{\spM(S)}}
\newpii{
In the turn from the binary to the multiary update propagation, the work by \tral\ \cite{frankT-icmt15,frankT-icmt16,frankT-icmt17} is the closest to ours conceptually and in part technically. We made several brief remarks in the introduction, and now can provide more details. 
}{1}{-4em}
\newpii{
 \tral\ begin with defining a class \mspS\ of diagram shapes (they say, bases), and a multimodel is a graph diagram of a certain shape $S\in\mspS$ 
in a adhesive category \spG\  intended to model attributed typed graphs and their morphisms, \ie, a multimodel is a functor \frar{M}{S}{\spG} with $S\in\mspS$. A multimodel update \frar{m}{M}{M'} is a span of natural transformations, \flar{m_d}{M}{\hat m}, \frar{m_a}{\hat m}{M'}, whose head $\hat m$ is a diagram of the same shape, and injective transformations $m_d$ and $m_a$ specify, resp., deletions and additions provided by update $m$.  For a fixed $S$, this gives us a non-trivial category of multimodels \tralmm\ to be compared with the categorified version of our \fwk\ based on category \Corrcat\ with arrows as described in \sectref{sec:future}.1. A major distinction between the two is that \Corrcat\ is a general model space category without any further restrictions, everything needed is provided by a family of boundary functors  \frar{\prtbf_i}{\Corrcat}{\spA_i}.  In contrast, although the class of categories \compr{\tralmm}{S\in\mspS} is broad enough to be practically interesting, 
it does not include some corr structures appearing in applications, \eg, such as in our Running example, or in paper \cite{me-mdi10-springer} focused on UML modelling, or in our paper \cite{me-ecmfa17}, in which a very general pattern for corrs is described as a partial span of graph \mor s, or, finally, in general rule-based \syncon\ engines as described in \cite{egyed-sac15}. An advantage of the abstract lens \fwk\ is that all these constructs are uniformly modelled as category \Corrcat. 
\\
Further distinctions appear when we consider how update propagation operations are defined. \tral\ consider the concurrent case (a truly  impressive achievement), which we did not approach yet. However, they consider only two basic laws (Stability and Correctness) borrowed from the binary case, while our repertoire of laws is richer. They also have a sort of amendment operation implemented by Consistency Creating Rules, but it is a purely local operation independent on any other model in the \mmodel. In contrast, our amendments are tightly related to other models and in this sense are (self) propagation operations. Finally, a major distinction is that in their setting, consistency restoration is non-deterministic while our lenses are deterministic. Overall, lenses appear as an abstract algebraic interface to update propagation that can be implemented by different ways, \eg, with GDG or by repair rules as in 
\cite{egyed-sac15,DBLP:conf/models/RabbiLYK15}. 
}{cont'd}{-5cm}
%

For the state-based lens setting, the work closest in spirit to the turn from binary to multiary lenses, is Stevens' paper \cite{stevens-models17}. 
Her and our goals are similar, but the technical realizations are different even besides the state- vs. delta-based opposition. Stevens works with restorers, which take a multimodel (in the state-based setting, just a tuple of models) presumably {\em inconsistent}, and restores consistency by changing some models in the tuple while keeping other models (from the {\em authority set}) unchanged. In contrast, lenses take a {\em consistent} multimodel {\em and} updates, and return a consistent multimodel and updates.  As we argued in \sectref{sec:example}, including updates into the input data of the restoration operation allows better ``tuning'' update propagation policies as the inherited uncertainty of consistency restoration is reduced.}
Another important difference is update amendments, which are not considered in \cite{stevens-models17} -- models in the authority set are intact. 
%
{Yet another distinction is how the multiary vs. binary issue is treated. Stevens provides several results for decomposing an n-ary relation $R\in \xstar{\spA}$ into binary relations $R_{ij}\subseteq \xob{\spA_i}\timm \xob{\spA_j}$ between the components. For us, a relation is inherently $n$-ary, \ie, a set $R$ of $n$-ary links endowed with an $n$-tuple of projections \frar{\partial_i}{R}{A_i} uniquely identifying links' boundaries. 
Thus, while Stevens considers ``binarization'' of a relation $R$ by a {\em chain} of binary relations over the ``perimeter'' $A_1...A_n$, we binarize  it via the corresponding {\em span} of (binary) mappings $(\prt_1,{...},\prt_n)$ (UML could call this process {\em reification}). Our (de)composition results demonstrate advantages of the span view. 

Discussion of several other works in the state-based world, notably by Macedo \etal\ \cite{DBLP:conf/edbt/MacedoCP14} can be found in \cite{stevens-models17}.
}
Several remarks about the related work on the Putput law have already been made in \sectref{sec:backgr}.3

\zdmoddok{
Now several remarks about the related work on algebraic laws for lenses.
Compositionality of update propagation with update composition, the famous {Putput} law, is an important but probably the most controversial amongst lens laws. 
A proper formulation of such a law has always been a stumbling block for the lens \fwk\ as {Putput} without restrictions does not hold while finding an appropriate guarding condition -- not too narrow to be practically usable and not too wide to ensure compositionality -- has been elusive (cf. 
\cite{foster07%
	,bpierce-popl11%
	,me-jot11,jr-bx12,me-jss15%
}).
A preliminary idea of a constrained Putput is discussed in \cite{jr-bx12} under the name of a {\em monotonic} Putput: compositionality is only required for two consecutive deletions or two consecutive insertions (hence the term monotonic), which is obviously a too weak requirement. The idea of constraining Putput based on some compatibility relations over consecutive updates is much better; it was proposed by Orejas \etal\ in \cite{orejas-bx13}, but they did not consider reflective  amendments (and only dealt with trivial correspondences being pairs of models). Integration of \Kupd-constrained Putput and reflective updates (in the general setting for correspondences) is a contribution of the present paper.  
\\
In contrast to Putput, invertibility laws for symmetric binary lenses did not get that much attention in the bx community.  The issue of strong vs. weak invertibility (= invertibility in the present paper) was discussed in detail in \cite{me-models11,frank-models11}. A version of weak invertibility for the asymmetric case ---the PutGetPut law---was considered (in a different context of injective editing) in several early bx work from Tokyo, \eg, \cite{hu-aplas04}. As we mentioned in Remark \ref{rema:invert}, what is called roundtripping in \cite{bpierce-popl11} is, in fact, an identity propagation law rather than an invertibility law (and thus invertibility for ssc-lenses was not defined at all).  
}{}{-10cm}

\section{Conclusion}
Multimodel \syncon\ is an important practical problem, which cannot be fully automated but even partial automation would be beneficial. A major problem in building an automatic support is uncertainty inherent to consistency restoration. In this regard, restoration via update propagation rather than immediate repairing of an inconsistent state of the multimodel has an essential advantage:  having the update causing inconsistency as an input for the restoration operation can guide the propagation policy and essentially reduce the uncertainty. We thus come to the scenario of multiple model \syncon\ via multi-directional update propagation. We have also argued that reflective propagation to the model whose change originated inconsistency is a reasonable feature of the scenario. 

We presented a mathematical \fwk\ for  \syncon\ scenarios as above based on a multiary generalization of binary symmetric delta lenses introduced earlier, 
and enriched it with reflective propagation and KPutput law ensuring compatibility of update propagation with update composition in a practically reasonable way (in contrast to the strong but unrealistic Putput). We have also defined several operations of composing multiary lenses in parallel and sequentially. \zdnewwok{Our lens composition results make the \fwk\  interesting for practical applications: if a tool builder has implemented a library of elementary \syncon\ modules based on lenses and, hence, ensuring basic laws for change propagation, then a complex module assembled from elementary lenses will automatically be a lens and thus also enjoys the basic laws. This allows one to avoid additional integration testing, which can essentially reduce the cost of \syncon\ software.}{this piece is just copy-pasted from the intro}{-4em} 

\medskip
\noindent
{\bf Acknowledgement.} We are really grateful to anonymous reviewers for careful reading of the manuscript and detailed, pointed, and stimulating reviews, which essentially improved the presentation and discussion of the \fwk.

\bibliographystyle{splncs03} 
\bibliography{refsGrand17}
\end{document}